%% file: cluster_trees_sosa_format.tex
\documentclass[twoside,leqno,twocolumn]{article}

\usepackage[letterpaper]{geometry}

\usepackage{ltexpprt}

\usepackage{tikz}
\usetikzlibrary{automata,positioning,arrows,shapes,fit}
\usepackage{caption}
\usepackage{subcaption}
\usepackage[linesnumbered,ruled,vlined,algo2e]{algorithm2e}
\usepackage{float}
\usepackage{enumitem}
\usepackage[sort&compress,numbers]{natbib}
\usepackage{balance}

\usepackage{array}   
\newcolumntype{L}{>{$}l<{$}} 
\newcolumntype{R}{>{$}r<{$}} 
\newcolumntype{C}{>{$}c<{$}} 

\usepackage{flushend}
\usepackage{amsmath}
\usepackage{amsfonts}
\newtheorem{observation}{Observation}[section]
\newtheorem{definition}[Definition]{Definition}

\makeatletter
\renewcommand{\paragraph}{%
  \@startsection{paragraph}{4}%
  {\z@}{0.8ex \@plus 1ex \@minus .2ex}{-1em}%
  {\normalfont\normalsize\bfseries}%
}
\makeatother

\title{A Breezing Proof of the KMW Bound}

\author{Corinna Coupette\thanks{MPI for Informatics; Saarbr\"{u}cken Graduate School of Computer Science; coupette@mpi-inf.mpg.de.}
	\and Christoph Lenzen\thanks{MPI for Informatics; clenzen@mpi-inf.mpg.de.}}

\date{}

\usepackage[hidelinks]{hyperref}

\makeatletter
\renewcommand*{\@opargbegintheorem}[3]{\trivlist
	\item[\hskip \labelsep{\sc #1\ #2.}] \textsc{(#3)}\ \itshape}
\makeatother

\newenvironment{xproof}[1]{\noindent\textit{Proof of #1.}}{\qquad\vbox{\hrule height0.6pt\hbox{%
			\vrule height1.3ex width0.6pt\hskip0.8ex
			\vrule width0.6pt}\hrule height0.6pt
	}\outerparskip 0pt\endtrivlist}

\begin{document}

\maketitle
	
	
	
	
	
	
	\pagenumbering{arabic}
	\setcounter{page}{1}
	\cfoot{\thepage}

\begin{abstract}
In their seminal paper from 2004, Kuhn, Moscibroda, and Wattenhofer (KMW) proved a hardness result for several fundamental graph problems in the LOCAL model: 
For any (randomized) algorithm, 
there are graphs with $n$ nodes and maximum degree $\Delta$ 
on which $\Omega(\min\{\sqrt{\log n/\log \log n},\log \Delta/\log \log \Delta\})$ (expected) communication rounds are required 
to obtain polylogarithmic approximations to a minimum vertex cover, 
minimum dominating set, or maximum matching.
Via reduction, this hardness extends to symmetry breaking tasks like finding maximal independent sets or maximal matchings.

Today, more than $15$ years later,
there is still no proof of this result that is easy on the reader. 
Setting out to change this, in this work, we provide a fully self-contained and $\mathit{simple}$ proof of the KMW lower bound.
Our key argument is algorithmic, and it relies on an invariant that can be readily verified from the generation rules of the lower bound graphs.
\end{abstract}

\input{intro}
\input{model}
\input{macro_level}
\input{micro_level}
\input{bounds}
\bibliographystyle{plain}
\balance
\bibliography{bibliography}

\appendix
\input{appendix}
\balance

\end{document}

%% file: intro.tex
\section{Introduction and Related Work}\label{sec:intro}

A key property governing the complexity of distributed graph problems is their \emph{locality}: 
the distance up to which the nodes running a distributed algorithm need to explore the graph to determine their local output. 
Under the assumption that nodes have unique identifiers, the locality of any task is at most $D$, the diameter of the graph. 
However, many problems of interest have locality $o(D)$, 
and understanding the locality of such problems in the LOCAL model of computation has been a main objective of the distributed computing community since the inception of the field.

A milestone in these efforts is the 2004 article by Kuhn, Moscibroda, and Wattenhofer, 
proving a lower bound of $\Omega(\min\{\sqrt{\log n/\log \log n}, \log \Delta/\log \log \Delta\})$ on the locality of several fundamental graph problems~\cite{kuhn2004}, 
where $n$ is the number of nodes and $\Delta$ is the maximum degree of the input graph. 
The bound holds under both randomization and approximation, 
and it is the first result of this generality
beyond the classic $\Omega(\log^* n)$ bound on $3$-coloring cycles \cite{linial1992}. %

\begin{figure*}[t]
	\centering
	\begin{subfigure}[b]{0.49\textwidth}
		\centering
		\input{tikz-ct1-hierarchy}
		\caption{Hierarchical representation}
	\end{subfigure} %
	\begin{subfigure}[b]{0.49\textwidth}
		\centering
		\input{tikz-ct1-flat}
		\caption{Flat representation}
	\end{subfigure}
	\caption[Representations of $CT_1$]{%
		Representations of $CT_1$, which is parametrized by $\beta$, shaded by cluster size (darker means smaller). %
		Cluster shapes indicate cluster position (\emph{internal} or \emph{leaf}). Edge label $i$ is short for $\beta^i$,
		the number of neighbors that nodes in one cluster have in another. 
		For example, nodes in cluster $C_0$ have $\beta^0$ neighbors in cluster $C_1$, and nodes in cluster $C_1$ have $\beta^1$ neighbors in cluster $C_0$.}\label{fig:ct1}
\end{figure*}
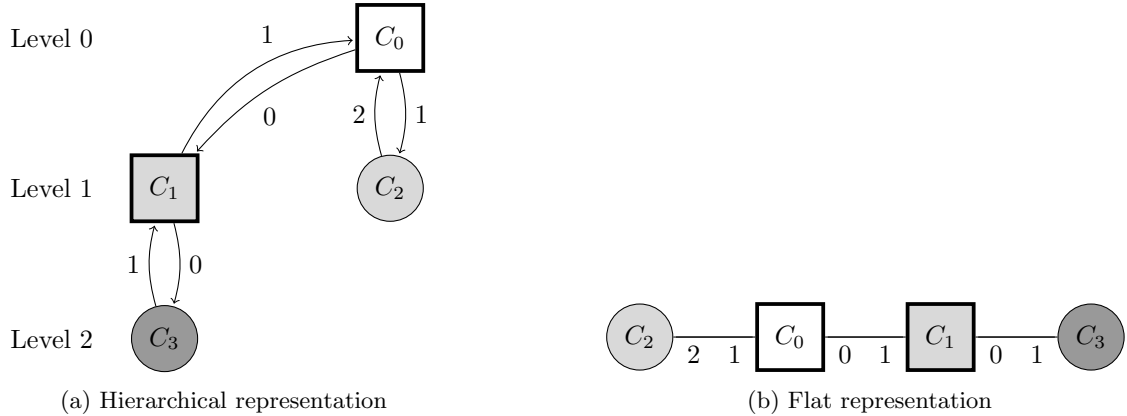

\subsection{A Brief Recap of the KMW Lower Bound}

In a nutshell, in~\cite{kuhn2016}, the authors reason as follows.

\begin{enumerate}[wide,itemsep=0pt]
	\item \textbf{Define Cluster Tree (CT) graph family.} 
	This graph family is designed such that in high-girth CT graphs, 
	the $k$-hop neighborhoods of many nodes that are \emph{not} part of a solution to,  
	e.g., the minimum vertex cover problem, 
	are isomorphic to the $k$-hop neighborhoods of nodes that \emph{are} part of a solution.\label{step:definition}
	\item \textbf{Prove that high-girth CT graphs have isomorphic node views.}
	If a CT graph $G_k$ has girth at least $2k+1$, 
	the isomorphisms mentioned in Step~\ref{step:definition} exist.
	This implies that a distributed algorithm running for $k$ rounds,
	which needs to determine the output at nodes based on their $k$-hop neighborhood,
	cannot distinguish between such nodes based on the graph topology.\label{step:isomorphism}
	\item \textbf{Show existence of high-girth CT graphs.}
	For each $k\in \mathbb{N}$,
	there exists a CT graph $G_k$ with girth at least $2k+1$ that has sufficiently few nodes and low maximum degree.\label{step:existence}
	\item \textbf{Infer lower bounds.}
	Under uniformly random node identifiers,%
	\footnote{In the LOCAL model,
		nodes have unique identifiers.
		Without these,
		even basic tasks like computing the size of the graph are impossible.}
	on a CT graph with girth at least $2k+1$,
	a $k$-round distributed algorithm cannot achieve a small expected approximation ratio for minimum vertex cover, maximum matching, or minimum dominating set,
	and it cannot find a maximal independent set or maximal matching with a small probability of failure.\label{step:bounds}
\end{enumerate}

The core of the technical argument lies in Step~\ref{step:isomorphism}.
A bird's-eye view of the reasoning for each of the steps is as follows.

\begin{enumerate}[wide,noitemsep]
	\item \textbf{Define Cluster Tree (CT) graph family.} 
	We want to have a large independent set of nodes%
	---referred to as \emph{cluster} $C_0$---%
	which contains most of the nodes in the graph.
	The $k$-hop neighborhoods of these nodes should be isomorphic not only to each other but also to the $k$-hop neighborhoods of nodes in a smaller cluster $C_1$.
	Each node in $C_0$ should have one neighbor in $C_1$,
	and the edges between the nodes from both clusters should form a biregular graph.
	In this situation, 
	a $k$-round distributed algorithm computing, e.g., a vertex cover, 
	cannot distinguish between the endpoints of edges connecting $C_0$ and $C_1$ 
	based on the graph topology.
	This is all we need for Step~\ref{step:bounds} to succeed.
	
	However, choosing the ratio $\beta := |C_0|/|C_1|$ larger than $1$ entails that nodes in $C_1$ have more neighbors in $C_0$ than vice versa.
	To maintain the indistinguishability of nodes in $C_0$ and $C_1$ for a $k$-round distributed algorithm,
	we add clusters $C_2$ and $C_3$ providing the ``right'' number of additional neighbors to $C_0$ and $C_1$, respectively,
	which are by a factor of $\beta$ smaller than their neighboring cluster to keep the overall number of non-$C_0$ nodes small.
	Now the nodes in $C_0$ and $C_1$ have the same number of neighbors,
	which implies that one round of communication is insufficient to distinguish between them.%
	\footnote{%
		This only applies if nodes do \emph{not} know the identities of their neighbors initially,
		known as KT0 (initial knowledge of topology up to distance $0$).
		It is common to assume KT1, i.e.,
		nodes \emph{do} know the identifiers of their neighbors at the start of the algorithm.
		However, this weakens the lower bound by one round only,
		not affecting the asymptotics.}
	See Figure~\ref{fig:ct1} (p.~\pageref{fig:ct1}) for an illustration of the resulting structure, $CT_1$.
	
	Unfortunately, looking up to distance two will now reveal the difference in degrees of neighbors:
	``Hiding'' the asymmetry between $C_0$ and $C_1$ by adding $C_2$ and $C_3$ enforces a similar asymmetry between $C_2$ and $C_3$.
	This is overcome by inductively ``growing'' a \emph{skeleton tree} structure on clusters,
	which encodes the topological requirements for moving the asymmetry in degrees further and further away from $C_0$ and $C_1$.
	
	Because in a graph of girth at least $2k+1$,
	the $k$-hop neighborhoods of all nodes are trees,
	the symmetry in degrees thus established is sufficient to result in isomorphic $k$-hop neighborhoods between nodes in $C_0$ and $C_1$.
	The growth rules of the skeleton tree are chosen to meet the topological requirements,
	while increasing degrees and the total number of nodes as little as possible.\medskip
	
	\item \textbf{Prove that high-girth CT graphs have isomorphic node views.}
	Using that $k$-hop neighborhoods of high-girth CT graphs are trees,
	the task of showing that $v\in C_0$ and $w\in C_1$ have isomorphic $k$-hop neighborhoods boils down to finding a \emph{degree-preserving bijection} between these neighborhoods that maps $v$ to $w$.
	At first glance,
	this seems straightforward:
	By construction, 
	nodes in inner clusters of the skeleton tree have degrees of $\beta^0,\beta^1,\ldots,\beta^k$ towards their $k+1$ adjacent clusters,
	and for each leaf cluster that lies at distance $d \leq k$ from $C(v)$ and has a degree of $\beta^x$ towards its parent cluster, 
	we can find a leaf cluster with the same degree towards its parent cluster at distance $d$ from $C(w)$.
	Hence, 
	mapping a node $v''$ with parent $v'$ to a node $w''$ with parent $w'$ if (1) the clusters $C(v')$ and $C(w')$ lie at the same distance $d'<k$ from $C(v)$ resp.~$C(w)$ and (2) $C(v')$ and $C(w')$ have the same outdegree towards $C(v'')$ resp.~$C(w'')$
	seems to be a promising approach for finding the desired bijection.
	
	However,
	when rooting the $k$-hop neighborhood of $v\in C_0$ ($w\in C_1$) at $v$ ($w$) and constructing the isomorphism by recursing on subtrees,
	for each processed node, the image of its parent under the isomorphism has already been determined.
	The asymmetry discussed in Step~\ref{step:definition} also shows up here:
	Some children of $v\in C_0$ and $w\in C_1$ that are mapped to each other will have different degrees towards their parents' clusters.
	This results in a mismatch for one pair of \emph{their} neighbors when processing a node according to the proposed strategy.
	
	Nonetheless,
	it turns out that mapping such ``mismatched'' nodes to each other results in the desired bijection.
	Proving this is,
	by a margin,
	the technically most challenging step in obtaining the KMW lower bound.\medskip

	\item \textbf{Show existence of high-girth CT graphs.}
	In order to show that sufficiently small and low-degree CT graphs $G_k$ of girth $2k+1$ exist,
	Kuhn et al.\ make use of \emph{graph lifts.}%
	\footnote{In the original paper~\cite{kuhn2004},
		they instead use subgraphs of a high-girth family of graphs $D(r,q)$ given in~\cite{lazebnik1995}.
		Utilizing lifts as outlined here was proposed by Mika G\"o\"os and greatly simplifies a self-contained presentation.}
	Graph $H$ is a lift of graph $G$
	if there exists a \emph{covering map} from $H$ to $G$,
	i.e., a surjective graph homomorphism that is bijective when restricted to the neighborhood of each node of $H$.
	These requirements are stringent enough to ensure that a lift of a CT graph $G_k$
	(i) is again a CT graph,
	(ii) has at least the same girth,
	and (iii) has the same maximum degree.
	On the other hand, they are lax enough to allow for \emph{increasing} the girth.%
	\footnote{For instance, the cycle $C_{3t}$ on $3t$ nodes is a lift of $C_3$,
		where the covering map sends the $i$\textsuperscript{th} node of $C_{3t}$ to the $(i \bmod 3)$\textsuperscript{th} node of $C_3$.
		Any graph $G$ has an acyclic lift that is an infinite tree $T$,
		by adding a new ``copy'' of node $v\in V(G)$ to $T$ for each walk leading to $v$ (when starting from an arbitrary fixed node of $G$ whose first copy is the root of $T$).
		The challenge lies in finding \emph{small} lifts of high girth.}
	This is exploited by a combination of several known results as follows.
	\begin{enumerate}[noitemsep]
		\item Construct a low-girth CT graph $G_k$ by connecting nodes in clusters that are adjacent in the skeleton tree using the edges of disjoint complete bipartite graphs whose dimensions are prescribed by the edge labels of the skeleton tree.%
		\footnote{%
		E.g., the nodes in clusters $C_0$ and $C_1$, which themselves are connected by an edge with labels $(\beta^0,\beta^1)$ (cf. Figure~\ref{fig:ct1}, p.~\pageref{fig:ct1}), 
		are connected using the edges of $|C_0|/\beta^1$ copies of $K_{\beta^0,\beta^1}$.
		} 
		Choose the smallest such $G_k$.
		\item Embed $G_k$ into a marginally larger regular graph,
		whose degree is the maximum degree of $G_k$ (this is a folklore result).
		\item There exist $\Delta$-regular graphs of girth $g$ and fewer than $\Delta^g$ nodes~\cite{erdos1963}.
		\item For any two $\Delta$-regular graphs of $n_1$ and $n_2$ nodes,
		there is a common lift with $O(n_1 n_2)$ nodes~\cite{angluin1981}.
		Apply this to the above two graphs to obtain a high-girth lift of a supergraph of $G_k$.
		\item Restrict the covering map of this lift to the preimage of $G_k$ to obtain a high-girth lift of $G_k$,
		which itself is a CT graph.
	\end{enumerate}
	Doing the bookkeeping yields size and degree bounds for the obtained CT graph as a function of $k$.\medskip

	\item \textbf{Infer lower bounds.}
	With the first three steps complete,
	the lower bound on the number of rounds for minimum vertex cover approximations follows by showing that the inability to distinguish nodes in $C_0$ and $C_1$ forces the algorithm to choose a large fraction of nodes from $C_0$,
	while a much smaller vertex cover exists.
	The former holds because under a uniformly random labeling,
	nodes in $C_0$ and $C_1$ are equally likely to be selected,
	while each edge needs to be covered with probability $1$.
	Thus, at least $|C_0|/2$ nodes are selected in expectation.
	At the same time,
	the CT graph construction ensures that $C_0$ contributes the vast majority of the nodes.
	Hence, choosing all nodes \emph{but} the independent set $C_0$ results in a vertex cover much smaller than $|C_0|/2$.
	The lower bounds for other tasks follow by similar arguments and reductions.%
	\footnote{%
		For example, as any maximal matching yields a $2$-approxi\-mation to a minimum vertex cover,
		the minimum vertex cover lower bound extends to maximal matching.}
\end{enumerate}

\subsection{Our Contribution}

Despite its significance, apart from an early extension to maximum matching by the same authors \cite{kuhn2006}, 
the KMW lower bound has not inspired follow-up results. 
We believe that one reason for this is that the result is not as well-understood as the construction by Linial~\cite{linial1992},
which inspired many extensions~\cite{naor1995,lenzen2008,czygrinow2008,goos2013,goos2017,brandt2016,chang2018,balliu2019} and alternative proofs~\cite{laurinharju2014,suomela2019}. 
History itself appears to drive this point home: 
In a 2010 arXiv article~\cite{kuhn2010}, 
an improvement to $\Omega(\min\{\sqrt{\log n},\log \Delta\})$ was claimed,
which was refuted in 2016 by Bar-Yehuda et al.~\cite{baryehuda2016}.
2016 was also the year when finally a journal article covering the lower bound was published~\cite{kuhn2016}%
---over a decade after the initial construction! 
In the journal article,
the technical core of the proof spans six pages,
involves convoluted notation,
and its presentation suffers from a number of minor errors impeding the reader.%
\footnote{The refutation of the improved lower bound in~\cite{baryehuda2016} came to the attention of the authors of~\cite{kuhn2016} \emph{after} the article had been accepted by \emph{J. ACM} with the incorrect result;
the authors were forced to revise the article on short notice before publication,
leading to the corrected material receiving no review~\cite{kuhn2020}.
Taking into account the complexity of the proof in~\cite{kuhn2016},
despite minor flaws, 
we feel that the authors did a commendable job.
} 

\paragraph{A constructive proof of the key graph isomorphism.}

In this work, we present a novel proof for Step~\ref{step:isomorphism} of the KMW bound. 
That is, we revise the heart of the argument, which
shows that nodes in $C_0$ and $C_1$ have indistinguishable $k$-hop neighborhoods. 
The proof in \cite{kuhn2016} uses an inductive argument that is based on a number of notation-heavy derivation rules to describe the $k$-hop neighborhoods of nodes in $C_0$ and $C_1$ and map subtrees of these neighborhoods onto each other. 
The proofs of the derivation rules, which together enable the inductive argument, 
rely crucially on notation and verbal description.

In contrast,
our proof is based on a simple algorithmic invariant. 
We give an algorithm that constructs the graph isomorphism between the nodes' neighborhoods in the natural way suggested by the CT graph construction.
The key observation is that one succinct invariant is sufficient to overcome the main obstacle, 
namely the ``mismatched'' nodes that are mapped to each other by the constructed isomorphism.
This not only substantially simplifies the core of the proof,
it also has explanatory power:
In the proof from~\cite{kuhn2016},
the underlying intuition is buried under heavy notation and numerous indices.

\paragraph{Simplified notation and improved visualization.}
Capitalizing on the new proof of the key graph isomorphism,
as a secondary contribution, 
we clean up and simplify notation also outside of the indistinguishability argument.
We complement this effort with improved visualizations of the utilized graph structures.
Overall, we expect these modifications to make the lower bound proof much more accessible,
and we hope to provide a solid foundation for work extending the KMW result.

\subsection{Further Related Work}

The KMW bound applies to fundamental graph problems that are locally checkable in the sense of Naor and Stockmeyer \cite{naor1995}.
Balliu et al.\ give an overview of the known time complexity classes for such problems \cite{balliu2018,balliu2018b,balliu2020}, 
extending a number of prior works \cite{fraigniaud2013,chang2016,feuilloley2016,ghaffari2017,ghaffari2018,chang2019b,chang2019,rozhon2020}, 
and Suomela surveys the state of the art attainable via constant-time algorithms \cite{suomela2013}. 
Bar-Yehuda et al.\ provide algorithms that compute $(2+\varepsilon)$-approximations to minimum (weighted) vertex cover and maximum (weighted) matching in $\mathcal{O}(\log \Delta / \varepsilon \log \log \Delta)$ and $\mathcal{O}(\log \Delta / \log \log \Delta)$ deterministic rounds, respectively \cite{baryehuda2016,baryehuda2017}, 
demonstrating that the KMW bound is tight when parametrized by $\Delta$ even for constant approximation ratios. 
For symmetry breaking tasks,
the classic algorithm by Panconesi and Rizzi \cite{panconesi2001} to compute maximal matchings and maximal independent sets in $\mathcal{O}(\log^*n+\Delta)$ deterministic rounds has recently been shown to be optimal for a wide range of parameters \cite{balliu2019}. 

\subsection{Organization of this Article}

This article gives a complete and self-contained proof of the KMW bound, 
supplementing the version focusing on the indistinguishability argument published at SOSA 2021.
After introducing basic graph theoretical concepts and notation as well as our computational model in Section~\ref{sec:prelim}, 
we define the lower bound graphs in Section~\ref{macro:definition}. 
This sets the stage for our main contribution:
In Section~\ref{macro:indistinguishability}, we prove the indistinguishability of the $k$-hop neighborhoods of nodes in the clusters $C_0$ and $C_1$ under the assumption of high girth. 

We infer the order and maximum degree of the lower bound graphs, 
which play important roles in the lower bound derivation, in Section~\ref{macro:properties}.
To ensure that lower bound graphs with high girth exist, we construct such graphs with low girth in Section~\ref{micro:low-girth} and lift them to high girth in Section~\ref{micro:high-girth} with the help of regular graphs introduced in Section~\ref{micro:regular}.
We obtain the KMW bound for polylogarithmic approximations to a minimum vertex cover in Section~\ref{sec:bounds}. 
The appendix provides extensions to minimum dominating set, maximum matching, maximal matching, and maximal independent set.

%% file: tikz-ct1-hierarchy.tex
		\begin{tikzpicture}[shorten >=1pt,node distance=2cm and 3cm,on grid,auto] 
		\node[state] (c00) [align=center,line width=0.5mm,rectangle] {$C_0$};
		\node[] (level0) [left=4.5cm of c00] {Level $0$};
		\node[] (level1) [below= of level0] {Level $1$};
		\node[] (level2) [below= of level1] {Level $2$};
		\node[state,fill=gray!30] (c01) [align=center,line width=0.5mm,rectangle, below left=of c00] {$C_1$};
		\node[state,fill=gray!30] (c02) [align=center, below=of c00] {$C_2$};
		\node[state,fill=gray!80] (c03) [align=center, below=of c01] {$C_3$};
		\path[->]
		(c01) edge [bend left=30, pos=0.6, align=left] node[below=0.5] {$0$} (c00)
		(c00) edge [bend right=15, align=left] node[above=0.5] {$1$} (c01)
		(c00) edge [bend left=15, align=left] node {$1$} (c02)
		(c02) edge [bend left=15, align=left] node {$2$} (c00)
		(c01) edge [bend left=15, align=left] node {$0$} (c03)
		(c03) edge [bend left=15, align=left] node {$1$} (c01)
		;
	\end{tikzpicture}

%% file: tikz-ct1-flat.tex
\begin{tikzpicture}[shorten >=1pt,node distance=1cm and 2cm,on grid,auto] 
\node[state] (c00) [align=center,line width=0.5mm,rectangle] {$C_0$};
\node[state] (c01) [align=center,fill=gray!30,line width=0.5mm,rectangle, right=of c00] {$C_1$};
\node[state] (c02) [align=center,fill=gray!30, left=of c00] {$C_2$};
\node[state] (c03) [align=center,fill=gray!80, right=of c01] {$C_3$};
\path
	(c00) edge [pos=0.25, align=left,below] node {$0$} (c01)
	(c01) edge [pos=0.25, align=left,below] node {$1$} (c00)
	(c00) edge [pos=0.25, align=left,below] node {$1$} (c02)
	(c02) edge [pos=0.25, align=left,below] node {$2$} (c00)
	(c01) edge [pos=0.25, align=left,below] node {$0$} (c03)
	(c03) edge [pos=0.25, align=left,below] node {$1$} (c01)
	;
\end{tikzpicture}

%% file: model.tex
\input{notation-table}

\section{Preliminaries}\label{sec:prelim}
The basic graph theoretic notation used in this work is summarized in Table~\ref{tab:notation} (p.~\pageref{tab:notation}); 
all our graphs are finite and simple.

We operate in the LOCAL model of computation, our presentation of which follows Peleg \cite{peleg2000}.
The LOCAL model is a stylized model of network communication designed to capture the locality of distributed computing. 
In this model, a communication network is abstracted as a simple graph $G=(V,E)$, 
with nodes representing network devices and edges representing bidirectional communication links.  
To eliminate all computability restrictions that are not related to locality, the model makes the following assumptions:

\begin{itemize}[label=--, noitemsep]
	\item Network devices have unique identifiers and unlimited computation power.
	\item Communication links have infinite capacity.
	\item Computation and communication takes place in synchronous rounds.
	\item All network devices start computing and communicating at the same time. 
	\item There are no faults. 
\end{itemize}

In each round, a node can 

\begin{enumerate}[noitemsep]
	\item perform an internal computation based on its currently available information, 
	\item send messages to its neighbors, 
	\item receive all messages sent by its neighbors, and
	\item potentially terminate with some local output.
\end{enumerate}

A $k$-round distributed algorithm in the LOCAL model can be interpreted as a function from $k$-hop subgraphs to local outputs:
\begin{definition}[$k$-round distributed algorithm]\label{def:distributed-algorithm}
	A \emph{$k$-round distributed algorithm} $\mathcal{A}$ is a function mapping $k$-hop subgraphs $G^k(v)$, 
	labeled by unique node identifiers (and potentially some local input), 
	to local outputs.
	For a randomized algorithm, nodes are also labeled by (sufficiently long) strings of independent, unbiased random bits.
\end{definition}

We assume that at the start of the algorithm, nodes do \emph{not} know their incident edges.
Assuming that nodes \emph{do} know these edges in the beginning weakens the lower bound by one round only, not affecting the asymptotics.

The key concept used to show that a graph problem is difficult to solve (exactly or approximately) for a $k$-round distributed algorithm in the LOCAL model is the \emph{$k$-hop indistinguishability} of nodes' neighborhoods.%
\footnote{%
	LOCAL algorithms may also make use of nodes' local inputs and identifiers.
	However, so far, the KMW construction has been applied to tasks without additional inputs only.
	For such tasks,
	assigning node identifiers uniformly at random translates the stated purely topological notion of indistinguishability to identical \emph{distributions} of $k$-hop subgraphs labeled by identifiers.
}
\begin{definition}[$k$-hop indistinguishability in $G$]\label{def:indistinguishability} 
	Two~nodes $v$ and $w$ in $G$ are indistinguishable to a $k$-round distributed algorithm (\emph{$k$-hop indistinguishable}) 
	if and only if there exists an isomorphism $\phi: V(G^k(v))\rightarrow V(G^k(w))$ with $\phi(v) = w$. 
\end{definition}
Accordingly, our goal in Section~\ref{macro:indistinguishability} will be to establish that the nodes in $C_0$ and $C_1$ are $k$-hop indistinguishable.

%% file: notation-table.tex
\begin{table*}[thb!]
	\setlength{\tabcolsep}{1.5pt}
	\centering\footnotesize\renewcommand{\arraystretch}{1.1}
	\begin{tabular}{RCp{0.35\textwidth}p{0.46\textwidth}}\hline
		\mathbf{Symbol}&&\textbf{Definition}&\textbf{Meaning}\\\cline{1-4}
		[k]&:=&$\{i \in \mathbb{N} \mid i \leq k\}$&Set of positive integers not larger than $k$\\
		\ [k]_0&:=&$\{i \in \mathbb{N}_0 \mid i \leq k\}$&Set of nonnegative integers not larger than $k$\\
		\cline{1-4}
		G&:=&$(V(G), E(G))$&Graph with node set $V(G)$ and edge set $E(G)$\\
		G[S]&:=&$(S,\{ \{ v,w \} \in E(G)\mid v,w\in S \}$%
		&Subgraph of $G$ induced by $S\subseteq V(G)$\\
		\Gamma_G(v)&:=&$\{w \in V(G) \mid \{v,w\}\in E(G)\}$&Neighborhood of $v$ in $G$ (non-inclusive)\\
		\Gamma_G^k(v)&:=&$\{w\in V(G)\mid d_G(v,w)\leq k \}$&$k$-hop neighborhood of $v$ in $G$ (inclusive)\\
		G^k(v)&:=&$G[\Gamma^k(v)]\setminus \{\{w,u\} \in E(G)$\newline $\mid 
		d_G(v,w)=d_G(v,u)= k\}$&$k$-hop subgraph of a node $v$ in $G$\\\hline
		\exists p_G(u,w,k)&:=&$\exists(\{u,v_1\}, \{v_1,v_2\}, \dots, \{v_{k-1}, w\})$\newline $\in E(G)^k$
		&Existence of $k$-hop path from $u$ to $w$ in $G$\\
		d_G(u,w)&:=&$\min \{k \mid \exists p_G(u,w,k)\}$&Distance between node $u$ and node $w$ in $G$\\
		g_G&:=&$\inf \{k > 0\mid \exists v\in V(G), p_G(v,v,k)\}$&Girth of $G$ (length of its shortest cycle)\\\hline
	\end{tabular}
	\caption[General Notation]{General notation used in this work (subscript or parenthesized $G$ may be omitted when clear from context).}\label{tab:notation}
	\renewcommand{\arraystretch}{1.0}
\end{table*}

%% file: macro_level.tex
\section{Cluster Trees}\label{sec:macro_level}
Cluster Trees (CTs) are the main concept in the derivation of the KMW bound.
For $k\in\mathbb{N}$, 
the \emph{Cluster Tree skeleton} $CT_k$ describes sufficient constraints on the topology of graphs $G_k$ (beyond high girth, which ensures that the $k$-hop neighborhoods of all nodes are trees) to enable the indistinguishability proof in Section~\ref{macro:indistinguishability}.
\begin{definition}[Cluster Trees]\label{def:ct-skeleton}
	For $k\in\mathbb{N}$, a cluster tree skeleton (\emph{CT skeleton}) is a tree $CT_k = (\mathcal{C}_k, \mathcal{A}_k)$,
	rooted at $C_0\in \mathcal{C}_k$,
	which describes constraints imposed on a corresponding \emph{CT graph} $G_k$.
	\begin{itemize}[noitemsep,label=--]
	  \item For each \emph{cluster}%
	  \footnote{``Cluster'' here is used in the sense of ``associated set of nodes,''
	  referring to its role in $G_k$.
	  We use the term to refer to both nodes in $CT_k$ and the corresponding independent sets in $G_k$.}
	  $C\in\mathcal{C}_k$,
	  there is a corresponding independent set in $G_k$.
	  \item An edge connecting clusters $C$ and $C'$ in $CT_k$ is labeled with $\{(C,x),(C',y)\}$ for $x,y\in \mathbb{N}$.
	  This expresses the constraint that in $G_k$, 
	  $C$ and $C'$ must be connected as a biregular bipartite graph,
	  where each node in $C$ has $x$ neighbors in $C'$ and each node in $C'$ has $y$ neighbors in $C$.
	  We say that $C$ $(C')$ is connected to $C'$ $(C)$ via \emph{outgoing label} $x$ $(y)$.
	  \item $G_k$ contains no further nodes or edges.
	\end{itemize}
\end{definition}
Note that $CT_k$ imposes many constraints on $G_k$.
Choosing the size of $C_0$ determines the number of nodes and edges in $G_k$,
and node degrees are fully determined by $CT_k$ as well.
However, there is substantial freedom regarding how to realize the connections between adjacent clusters.
As mentioned earlier,
this permits leveraging graph lifts to obtain cluster tree graphs $G_k$ of high girth in Step~\ref{step:existence} of the KMW construction.

\subsection{Construction of Cluster Tree Skeletons}\label{macro:definition}

Definition~\ref{def:ct-skeleton} (p.~\pageref{def:ct-skeleton}) does not detail the structure of $CT_k$. 
To specify this structure, we use the following terminology.

\begin{figure*}[hbt!]
	\centering
	\begin{subfigure}[b]{0.95\textwidth}
		\centering
		\input{tikz-ct2-flat-clustertypes}
		\caption{Flat representation of $CT_2$}
	\end{subfigure}
	\begin{subfigure}[b]{0.95\textwidth}
		\centering
		\input{tikz-ct3-flat-clustertypes}
		\caption{Flat representation of $CT_3$}
	\end{subfigure}
	\caption[Representations of $CT_2$ and $CT_3$]{Representations of $CT_2$ and $CT_3$, colored by cluster types; grey: internal, black: first growth rule, green: second growth rule.}\label{fig:ct2}
\end{figure*}
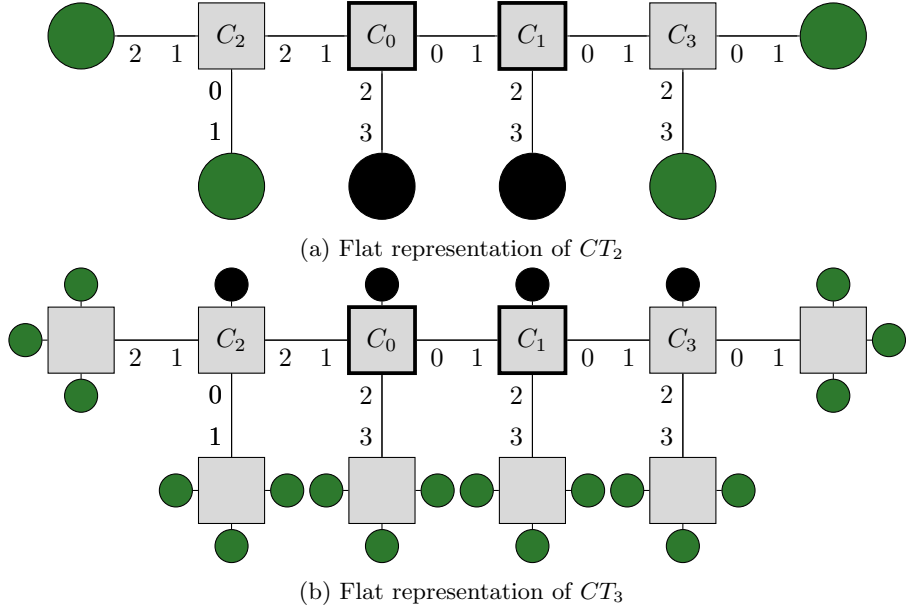

\begin{definition}[Cluster position, level, and parent]\label{def:cluster-position}
	A leaf cluster $C$ in $CT_k$ has position \emph{leaf,}
	while internal clusters have position \emph{internal}.
	The \emph{level} of $C$ is its distance to $C_0$.
	The \emph{parent cluster} of $C\neq C_0$ is its parent in $CT_k$.
\end{definition}

Given $\beta \geq 2(k+1)$,
the structure of $CT_k$ is now defined inductively.
The base case of the construction is $CT_1$.

\begin{definition}[Base case $CT_1$]\label{def:ct-trunk}
	$CT_1 = (\mathcal{C}_1,\mathcal{A}_1)$, where $\mathcal{C}_1 := \{C_0,C_1,C_2,C_3\}$ and
	\begin{align*}
	\mathcal{A}_1 := &\{\{(C_0,\beta^0),(C_1, \beta^1)\}, \{(C_0,\beta^1),(C_2, \beta^2)\},\\ 
	&\{(C_1,\beta^0),(C_3, \beta^1)\}\}.
	\end{align*}
\end{definition}  

Based on $CT_{k-1}$, for $k\geq 2$, $CT_k$ is grown as follows.
\begin{definition}[Growth rules for $CT_k$ given $CT_{k-1}$]\label{def:ct-growth}
	\begin{enumerate}[noitemsep,topsep=0em]
		\item To each internal cluster $C$ in $CT_{k-1}$, attach a new neighboring cluster $C'$ via an edge $\{ (C,\beta^k), (C',\beta^{k+1})\}$.\label{enum:ctgrowth:rule1}
		\item To each leaf cluster $C$ in $CT_{k-1}$ that is connected to its parent cluster via outgoing label $\beta^q$, 
		add a total of $k$ neighboring clusters: 
		one cluster $C'$ with an edge $\{ (C,\beta^p), (C',\beta^{p+1})\}$ for each $p\in [k]_0\setminus \{q\}$.\label{enum:ctgrowth:rule2}
	\end{enumerate}
\end{definition}

Note that with this definition, $CT_k$ is a regular tree but a CT graph $G_k$ is not regular.
Figure~\ref{fig:ct1} (p.~\pageref{fig:ct1}) shows $CT_1$ in its hierarchical and flat representations, 
and flat representations of $CT_2$ and $CT_3$ are given in Figure~\ref{fig:ct2} (p.~\pageref{fig:ct2}) to illustrate the growth process.\footnote{%
	The labels of the edges connecting leaf clusters in $CT_3$ to the rest of $CT_3$ are omitted in the drawing. 
	They are such that every internal cluster has outgoing labels $\{\beta^i\mid i\in [3]_0\}$, 
	and if a leaf cluster $C$ is connected to an internal cluster $C'$ with label $\beta^i$ outgoing from $C'$, then $C$ has outgoing label $\beta^{i+1}$.
}
In all figures, we write $i$ for outgoing label $\beta^i$ to reduce visual clutter,
and in the flat representations, 
outgoing labels are depicted like port numbers,
i.e., the edge label corresponding to $C$ is depicted next to $C$.

\subsection{Indistinguishability given High Girth}\label{macro:indistinguishability}

As observed by Kuhn et al.\ \cite{kuhn2004,kuhn2016}, 
showing $k$-hop indistinguishability becomes easier when the nodes' $k$-hop subgraphs 
are trees, i.e., the girth is at least $2k+1$. 
Notably, in a CT graph $G_k$ with $g\geq 2k+1$, 
the topology of a node's $k$-hop subgraph is determined entirely by the structure of the skeleton $CT_k$. 
Hence, we will be able to establish the following theorem without knowing the details of $G_k$.

\input{algorithm-indistuinguishability}

\begin{theorem}[$k$-hop indistinguishability of nodes in $C_0$ and $C_1$]\label{thm:indistinguishability}
	Let $G_k$ be a CT graph of girth $g\geq 2k+1$. 
	Then any $v_0\in C_0$ and $v_1\in C_1$ are $k$-hop indistinguishable.
\end{theorem}

By Definition~\ref{def:indistinguishability} (p.~\pageref{def:indistinguishability}), 
$v_0\in C_0$ and $v_1\in C_1$ are $k$-hop indistinguishable 
if and only if there exists an isomorphism $\phi: V(G_k^k(v_0))~\rightarrow~V(G_k^k(v_1))$ with $\phi(v_0) = v_1$.
We prove the theorem constructively by showing the correctness of Algorithm~\ref{alg:isomorphism} (p.~\pageref{alg:isomorphism}), which purports to find such an isomorphism.

Algorithm~\ref{alg:isomorphism} (p.~\pageref{alg:isomorphism}) implements a \emph{coupled depth-first search} on the $k$-hop subgraphs of $v_0\in C_0$ and $v_1\in C_1$.
Its main function, \textsc{FindIsomorphism($G_k, k, v_0, v_1$)}, receives a CT graph $G_k$ with high girth, 
along with the parameter $k$, 
and one node from each of $C_0$ and $C_1$ as input, and it outputs the $\phi$ we are looking for.
To obtain $\phi$, \textsc{FindIsomorphism} maps $v_0$ to $v_1$ and then calls the function \textsc{Walk($v_0$, $v_1$, $\bot$, $k$)} before it returns $\phi$.
The \textsc{Walk} function extends $\phi$ by mapping the \emph{newly discovered} nodes in the neighborhoods of its first two input parameters ($v$ and $w:=\phi(v)$, initially: $v_0$ and $v_1$) to each other with the help of the function \textsc{Map}. 
The third parameter of \textsc{Walk} ($prev$, initially: $\bot$) ensures that we only define $\phi$ for newly discovered nodes, 
while the fourth parameter ($depth$, initially: $k$) controls termination when \textsc{Walk} calls itself recursively on the newly discovered neighbors (and the newly discovered neighbors of these neighbors, and so on) until the entire $k$-hop subgraph of $v_0$ has been visited. 

The tricky part now is to ascertain that the interplay between the functions \textsc{Walk} and \textsc{Map} makes $\phi$ a bijection from $V(G_k^k(v_0))$ to $V(G_k^k(v_1))$,
i.e., nodes that are paired up always have the same degree.
Here, the representation of node neighborhoods used by the \textsc{Walk} function is key, 
which is based on the insight that the set of nodes neighboring $v$ (resp. $w$) can be partitioned by the outgoing labels in $CT_k$ through which neighboring nodes are discovered from $v$ ($w$).
Since these labels lie in $\{\beta^i \mid i\in[k+1]_0\}$, 
\textsc{Walk} represents the neighborhood of $v$ ($w$) as a list $N_v$ ($N_w$) of $k+2$ (possibly empty) lists (Algorithm~\ref{alg:isomorphism}, l.~\ref{alg:iso:neighborhood-start}--\ref{alg:iso:neighborhood-end}, p.~\pageref{alg:isomorphism}). 
The list at index $i$ holds all \emph{previously undiscovered} nodes 
(we require $v'\neq prev$ and $w'\neq \phi(prev)$) connected to $v$ ($w$) via $v$'s ($w$'s) outgoing label $\beta^i$, 
in arbitrary order.

The \textsc{Walk} function passes $N_v$ and $N_w$ to the function \textsc{Map} (Algorithm~\ref{alg:isomorphism}, l.~\ref{alg:iso:call-map}, p.~\pageref{alg:isomorphism}), 
which sets $\phi(N_v[i][j]) := N_w[i][j]$ where possible (Algorithm~\ref{alg:isomorphism}, l.~\ref{alg:iso:map-outer-for-start}--\ref{alg:iso:map-inner-for-end}, p.~\pageref{alg:isomorphism}). 
It then treats the special case that some nodes in $N_v$ and $N_w$ remain unmatched (Algorithm~\ref{alg:isomorphism}, l.~\ref{alg:iso:special-start}--\ref{alg:iso:special-end}, p.~\pageref{alg:isomorphism}). 
By construction, without this special case, 
the $\phi$ returned by \textsc{FindIsomorphism} is already an isomorphism between the subgraphs of $G_k^k(v_0)$ and $G_k^k(v_1)$ induced by the nodes of the domain for which $\phi$ is defined (and their images under $\phi$).
However, we still need to show that our special case suffices to extend this restricted isomorphism to a full isomorphism between $G_k^k(v_0)$ and $G_k^k(v_1)$. 
To facilitate our reasoning, we introduce \emph{cluster identities}:

\begin{definition}[Cluster identity $C(v)$]\label{def:cluster-identity}
	Given a node $v$ in a CT graph $G_k$, we refer to its cluster in $CT_k$ as its \emph{cluster identity}, denoted as $C(v)$.
	For example, for $v_0\in C_0$ and $v_1\in C_1$, we have $C(v_0) = C_0$, $C(v_1) = C_1$, and $C(v_0)\neq C(v_1)$. 
\end{definition}

Our argument will crucially rely on the concepts of \emph{node position} and \emph{node history}:

\begin{definition}[Node position]
	For $i\in\{0,1\}$, the \emph{position} of a node $v$ in $G_k^k(v_i)$ is the position of its cluster $C(v)$ in the CT skeleton (\emph{internal} or \emph{leaf}).
\end{definition}

\begin{definition}[Node history]
	For $i\in\{0,1\}$, the \emph{history} of a node $v\neq v_i$ in $G_k^k(v_i)$ is the outgoing label of the edge connecting $C(v)$ to $C(prev)$, i.e., $\beta^x$ if the corresponding edge is $\{(C(v),\beta^x),(C(prev),\beta^{x'})\}$.
\end{definition}

We begin with a simple observation:
\begin{lemma}[Variables determining node neighborhoods]\label{lem:position-history}
	For $v$ in $G_k^k(v_0)\setminus \{v_0\}$,
	let $w:=\phi(v)$.
	When \textsc{Map} is called with parameters $N_v$ and $N_w$ (Algorithm~\ref{alg:isomorphism} l.~\ref{alg:iso:call-map}, p.~\pageref{alg:isomorphism}), 
	the numbers of nodes in $N_v[i]$ and $N_w[i]$ for $i\in [k+1]_0$ are uniquely determined
	the \emph{position} and the \emph{history} of $v$ and $w$.
	If $v$ and $w$ agree on \emph{position} and \emph{history}, 
	$len(N_v[i])=len(N_w[i])$ for all $i\in [k+1]_0$.
	If $v$ and $w$ agree on position \emph{internal} but disagree on \emph{history}, 
	where $v$ has history $\beta^x$ and $w$ has history $\beta^y$, 
	we have $len(N_v[i])= len(N_w[i])$ for all $i\in [k+1]_0\setminus \{x,y\}$, $len(N_v[x])= len(N_w[x])-1$, and $len(N_v[y])-1= len(N_w[y])$.
\end{lemma}
\begin{proof}
	If $u\in \{v,w\}$ has position \emph{internal}, we know that $C(u)$ has outgoing labels $\{\beta^i\mid i\in [k]_0\}$ by the construction of the CT skeleton.
	Denoting by $z\in \{x,y\}$ the exponent of $u$'s \emph{history},
	we have that there are $\beta^i$ nodes in $N_u[i]$ for $i\in[k]_0\setminus \{z\}$, $\beta^z-1$ nodes in $N_u[z]$ (as $prev$ or $\phi(prev)$ are removed, respectively), and zero nodes in $N_u[k+1]$.

	If $u\in \{v,w\}$ has position \emph{leaf}, all nodes in $N_u$ belong to the same cluster $C'$,
	$u$ has $\beta^z$ neighbors in this cluster,
	and $prev$ (resp.~$\phi(prev)$) lies in this cluster as well.
	Hence, $len(N_u[z])=\beta^z-1$ and $len(N_u[i])=0$ for all $i\in [k+1]_0\setminus \{z\}$.

	From these observations,
	the claims of the lemma follow immediately.
\end{proof}

Using Lemma~\ref{lem:position-history} (p.~\pageref{lem:position-history}),
we can rephrase the task of proving Theorem~\ref{thm:indistinguishability} (p.~\pageref{thm:indistinguishability}) as a simple condition on the pairs of nodes on which \textsc{Walk} is called recursively.

\begin{corollary}[Sufficient condition for correctness of Algorithm~\ref{alg:isomorphism}]\label{cor:correctness-sufficient}
	Given a CT graph $G_k$ with girth at least $2k+1$, 
	if all pairs of nodes created by \textsc{Map} on which \textsc{Walk} is called recursively
	(i) agree on \emph{position} and \emph{history} 
	or (ii) agree on position \emph{internal}, 
	Algorithm~\ref{alg:isomorphism} (p.~\pageref{alg:isomorphism}) produces an isomorphism between $G^k_k(v_0)$ and $G^k_k(v_1)$.
\end{corollary}
\begin{proof}
	Due to the assumed high girth,
    Algorithm~\ref{alg:isomorphism} (p.~\pageref{alg:isomorphism}) produces an isomorphism between $G^k_k(v_0)$ and $G^k_k(v_1)$ 
	if $\phi{\big|}_{N_v}$ (i.e., $\phi$ with its domain restricted to the neighborhood of $v$) is a bijection from $N_v$ to $N_{\phi(v)}$ for all $v$ in $G^k_k(v_0)$ with $d(v,v_0)< k$.
	For $v_0$ and $\phi(v_0)=v_1$, this holds because they both have $\beta^i$ neighbors in the clusters connected to them via outgoing edge label $\beta^i$ for $i\in [k]_0$,
	i.e., $len(N_v[i])=len(N_w[i])$ for $i\in [k]_0$ (and $len(N_v[k+1]) = len(N_w[k+1]) = 0$). 
	Hence, \textsc{Map} ensures that $\phi(N_v)=N_w$.
	For nodes $v\neq v_0$ and $w:=\phi(v)$ paired by \textsc{Map} that agree on \emph{position} and \emph{history}, 
	Lemma~\ref{lem:position-history} (p.~\pageref{lem:position-history}) shows that $len(N_v[i])=len(N_w[i])$ for all $i\in [k+1]_0$,
	so again \textsc{Map} succeeds.
	The last case is that $v$ and $w$ agree on position \emph{internal}.
	In this case,
	applying Lemma~\ref{lem:position-history} (p.~\pageref{lem:position-history}) 
	and noting that \textsc{Map} takes care of the resulting mismatch in list lengths in Lines~\ref{alg:iso:special-start}--\ref{alg:iso:special-end}
	proves that \textsc{Map} succeeds here, too.
\end{proof}

The next step in our reasoning is to craft an algorithmic invariant establishing the preconditions of Corollary~\ref{cor:correctness-sufficient} (p.~\pageref{cor:correctness-sufficient}).
Reflecting the inductive construction of cluster trees,
we will prove it inductively.
To this end, for $i\in [k]$,
we interpret $CT_i$ as a subgraph of $CT_k$
by simply stripping away all clusters that were added after constructing $CT_i$.

Recall that $G_k^k(v_0)$ and $G_k^k(v_1)$ are trees, because the girth of $G_k$ is at least $2k+1$. 
Treating these trees as rooted at $v_0$ and $v_1$, respectively, 
Algorithm~\ref{alg:isomorphism} (p.~\pageref{alg:isomorphism}) maps nodes at depth $d$ in $G_k^k(v_0)$ to nodes at depth $d$ in $G_k^k(v_1)$.
Accordingly, the following notion will be useful.

\begin{definition}[Node parent]\label{def:parent-node}
	For $v \in G^k_k(v_i)$, $i\in \{0,1\}$, with $d(v_i,v) > 0$, the \emph{parent} of $v$ in $G^k_k(v_i)$, denoted $p_i(v)$,
	is the node through which $v$ is discovered from $v_i$ in Algorithm~\ref{alg:isomorphism}~(p.~\pageref{alg:isomorphism}).
\end{definition}

We are now ready to state the main invariant of Algorithm~\ref{alg:isomorphism} (p.~\pageref{alg:isomorphism}).

\begin{definition}[Main Invariant of Algorithm~\ref{alg:isomorphism}]\label{def:invariant}
	For $0 < d < k$, suppose that $v$ and $w:=\phi(v)$ lie at distance $d$ from $v_0$ and $v_1$, respectively.
	Then exactly one of the following holds:
	\begin{enumerate}[noitemsep]
		\item $C(v),C(w)\in CT_d$, and $v$ and $w$ agree on \emph{history} or
		both have \emph{history} $\leq \beta^{d+1}$.
		\item There is some $i$ with $d < i \leq k$ such that $C(v),C(w)\in CT_i\setminus CT_{i-1}$, 
		$v$ and $w$ agree on \emph{history}, and
		$C(v)$ and $C(w)$ are connected from $CT_{i-1}$ with outgoing labels $(\beta^{j'},\beta^{j'+1})$ for the same $j'\in[i]_0$. 
	\end{enumerate}
\end{definition}

\begin{figure*}[hbt!]
	\centering
	\begin{subfigure}{\textwidth}
		\centering
		\input{tikz-ct2-algorithm-3}
		\subcaption{$d = 1$ from $v_0$ and $v_1$: for the blue nodes, the first case of the invariant holds \emph{with} agreement on \emph{history}; 
			for the orange nodes, the first case of the invariant holds \emph{without} agreement on \emph{history}; 
			and for the green nodes, the second case of the invariant holds.}
	\end{subfigure}
	\begin{subfigure}{\textwidth}
		\centering
		\input{tikz-ct2-algorithm-5}
		\subcaption{$d = 2$ from orange nodes at distance $d = 1$: 
			because the invariant holds for $d = 1$, Corollary~\ref{cor:correctness-sufficient} (p.~\pageref{cor:correctness-sufficient}) ensures that Algorithm~\ref{alg:isomorphism} (p.~\pageref{alg:isomorphism}) produces an isomorphism between $G^2_2(v_0)$ and $G^2_2(v_1)$
			by mapping exactly one node in $G^2_2(v_0)$ discovered via the solid blue arrow to one node in $G^2_2(v_1)$ discovered via the solid green arrow (Algorithm~\ref{alg:isomorphism}, l.~\ref{alg:iso:special-start}--\ref{alg:iso:special-end}).
		}
	\end{subfigure}
	\caption{Illustration of Definition~\ref{def:invariant} (p.~\pageref{def:invariant}) for $CT_2$. 
		Cluster colors, shapes, and borders drawn as in Figure~\ref{fig:ct2} (p.~\pageref{fig:ct2}). 
		Nodes $v_0\in C_0$ and $v_1\in C_1$ are depicted as medium-size circles;
		representatives of nodes seen via a certain outgoing edge are depicted as small circles and connected to their parents by arrows.
		Node and arrow colors show outgoing edge labels (e.g., blue nodes are seen via the outgoing edge $\beta^0$); 
		dashed arrows indicate that $\beta^i-1$, rather than $\beta^i$, nodes are discovered via the outgoing label indicated by the arrow color.}\label{fig:invariant}
\end{figure*}
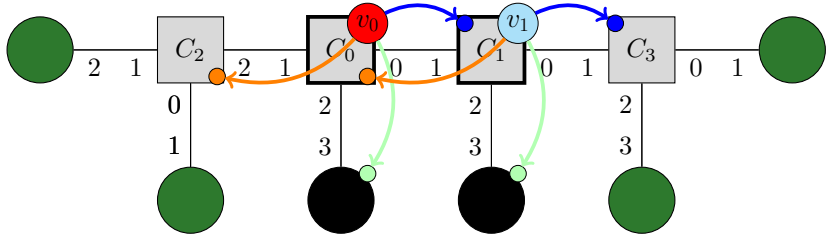
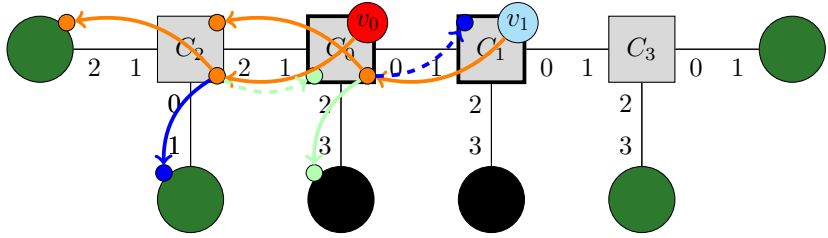

Intuitively, the first case tracks how asymmetry propagates through the construction,
and counts down how many levels of the iterative construction remain that hide it:
If there is a mismatch in history,
it is confined to $CT_d$,
i.e., not only $C(v),C(w)\in CT_d$,
but also the history of $v$ and $w$ has labels corresponding to $CT_d$.

However, we need to take into account that attaching leaf clusters to internal clusters by the first growth rule is needed,
as otherwise ``older'' clusters could be easily recognized without having to inspect the far-away clusters added by the second growth rule.
These leaves then recursively sprout their own subtrees, which again sprout their own subtrees, and so on.
If 
the recursive construction of the isomorphism visits such clusters 
before reaching distance $k$ from root $C_0$ (resp.\ $C_1$),
by design, it enters subtrees of $CT_k$ that started to grow in the same iteration.
Thus, these subtrees are completely symmetric,
and they are entered with matching history.
This is captured by the second case of the invariant.
The intuition of the invariant and its interplay with Corollary~\ref{cor:correctness-sufficient} (p.~\pageref{cor:correctness-sufficient}) are illustrated in Figure~\ref{fig:invariant} (p.~\pageref{fig:invariant}).

Recall that the growth rules only attach leaves,
and do so for each cluster.
Hence, the clusters in $CT_{k-1}$ are exactly the internal clusters in $CT_k$,
while $CT_k\setminus CT_{k-1}$ contains all leaves.
Thus, in the first case of the invariant,
$v$ and $w$ agree on position \emph{internal},
and in the second case,
they agree on \emph{position} and \emph{history}.
Therefore, Theorem~\ref{thm:indistinguishability} (p.~\pageref{thm:indistinguishability}) follows from Corollary~\ref{cor:correctness-sufficient} (p.~\pageref{cor:correctness-sufficient}) once the invariant is established.

Having carved out the crucial properties of the construction in this invariant,
we can now complete the proof of the theorem with much less effort than in~\cite{kuhn2016}.

\begin{lemma}[Main invariant holds]\label{lem:invariant}
Algorithm~\ref{alg:isomorphism} (p.~\pageref{alg:isomorphism}) satisfies the invariant stated in Definition~\ref{def:invariant} (p.~\pageref{def:invariant}).
\end{lemma}
\begin{proof}
	We prove the claim for fixed $k$ by induction on $d$.
	For $v$ and $w:=\phi(v)$ at distance $d=1$ from $v_0=p_0(v)$ and $v_1=p_1(w)$, respectively,
	$v$ and $w$ are matched in the initial call to \textsc{Walk} with $v_0$ and $v_1$ as arguments.
	In this call, $len(N_{v_0}[i])=len(N_{v_1}[i])$ for all $i\in [k+1]_0$,
	i.e., only nodes corresponding to the same outgoing labels get matched.
	Inspecting $CT_1$ and taking into account the CT growth rules,
	we see that for $i\in \{0,1\}$, 
	the matched nodes lie in clusters that are present already in $CT_1$ and have outgoing labels of at most $\beta^2$
	(i.e., the first case of the invariant holds),
	while for $i>1=d$,
	both nodes lie in clusters from $CT_i\setminus CT_{i-1}$ with outgoing labels of $\beta^{i+1}$
	and their clusters are connected from $CT_{i-1}$ with outgoing labels $(\beta^{i},\beta^{i+1})$
	(i.e., the second case of the invariant holds).
	
	For the inductive step, assume that the invariant is established up to distance $d$ for $1\leq d<k-1$, 
	and consider $v$, $w:=\phi(v)$ at distance $d+1$ from $v_0$ and $v_1$, respectively.
	We apply the invariant to $v':=p_0(v)$ and $w':=p_1(w)$ and distinguish between its two cases.
	
	(1) Suppose that $C(v'),C(w')\in CT_d$, and $v'$ and $w'$ agree on \emph{history} or both have \emph{history} $\leq \beta^{d+1}$.
	As $d<k$,
	we know that $v'$ and $w'$ agree on position \emph{internal}.
	By Lemma~\ref{lem:position-history} (p.~\pageref{lem:position-history}), 
	the call to \textsc{Walk} on $v'$ and $w'$ thus satisfies that $len(N_{v'}[i])=len(N_{w'}[i])$ for all $i\in [k+1]_0\setminus \{j,j'\}$, where $\beta^j,\beta^{j'}$ for $j, j' \leq d+1$ are the histories of $v'$ and $w'$, respectively.
	If $C(v)\in CT_{d+1}$,
	Lemma~\ref{lem:position-history} (p.~\pageref{lem:position-history}) entails that $v\in N_{v'}[i]$ for some $i\le d+1$, 
	and \textsc{Walk} chooses $w=\phi(v)$ from $N_{w'}[i']$ for some $i'\le d+1$.
	Due to the CT growth rules, 
	since $C(v'),C(w')\in CT_d$, the incident edges of $C(v')$ and $C(w')$ with outgoing labels of at most $\beta^{d+1}$ lead to clusters in $CT_{d+1}$,
	and the history of nodes discovered by traversing these edges is at most $\beta^{d+2}$.
	Hence, if $C(v)\in CT_{d+1}$, it follows that the first case of the invariant holds for $v$ and $w$.
	If $C(v)\notin CT_{d+1}$, we have that $C(v)\in CT_i\setminus CT_{i-1}$ for some $i> d+1$,
	yielding $len(N_{v'}[i])=len(N_{w'}[i])$, 
	and thus, $w\in N_{w'}[i]$.
	As $C(v')$ and $C(w')$ are internal clusters in $CT_{d+1}$,
	we can conclude that both $C(v)$ and $C(w)$ have been added to the cluster tree in the $i$\textsuperscript{th} construction step using growth rule~1.
	Hence, we get that $C(v),C(w)\in CT_i\setminus CT_{i-1}$ with $v$ and $w$ agreeing on \emph{history} $\beta^{i+1}$,
	and since $C(v'), C(w')\in CT_d\subseteq CT_{i-1}$, 
	$C(v')$ and $C(w')$ are connected from $CT_{i-1}$ with outgoing labels $(\beta^i,\beta^{i+1})$, and
	the second case of the invariant holds for $v$ and $w$.
	
	(2) Assume that there is some $i$ with $d < i \leq k$ such that $C(v'),C(w')\in CT_i\setminus CT_{i-1}$, 
	$v'$ and $w'$ agree on \emph{history}, and
	$C(v')$ and $C(w')$ are connected from $CT_{i-1}$ with outgoing labels $(\beta^{j'},\beta^{j'+1})$ for the same $j'\in[i]_0$. 
	Since $C(v')$ and $C(w')$ were added in the same growth round, 
	$v'$ and $w'$ also agree on \emph{position},
	so $v\in N_{v'}[j]$ and $w\in N_{w'}[j]$ for the same $j\in [k+1]_0$ by Lemma~\ref{lem:position-history} (p.~\pageref{lem:position-history}), 
	and similarly, 
	as $C(v')$ and $C(w')$ are both added when forming $CT_i$ from $CT_{i-1}$ and connected from $CT_{i-1}$ with the same labels,
	$v$ and $w$ agree on \emph{history}.
	Hence, if $j\neq j'+1$,
	then $C(v),C(w)\in CT_{i'+1}\setminus CT_{i'}$ for some $i' \geq i+1$, $C(v)$ and $C(w)$ are connected from $CT_i$ with the same labels,
	and since $i+1>d+1$, 
	the second case of the invariant holds for $v$ and $w$.
	If $j=j'+1$,
	$C(v)=C(p_0(v'))$ and $C(w)=C(p_1(w'))$.
	As \textsc{Walk} mapped $v'$ to $w'$,
	we have that $p_0(v')$ was mapped to $\phi(p_0(v'))=p_1(w')$,
	where $p_0(v')$ and $p_1(w')$ lie at distance $d-1$ from $v_0$ and $v_1$, respectively.
	Applying the invariant to these nodes,
	the first case and the second case with $i' \leq d+1$ both imply that $C(v),C(w)\in CT_{d+1}$,
	establishing the first case of the invariant for $v$ and $w$.
	And if the second case applies with $i'>d+1$,
	then the second case of the invariant holds for $v$ and $w$.
\end{proof}

With this result, we can summarize:\\

\begin{xproof}{Theorem~\ref{thm:indistinguishability} (p.~\pageref{thm:indistinguishability})}
	Follows from the correctness of Algorithm~\ref{alg:isomorphism} (p.~\pageref{alg:isomorphism}) for CT graphs $G_k$ with girth $\geq 2k+1$, 
	established via Lemma~\ref{lem:invariant} (p.~\pageref{lem:invariant}) and Corollary~\ref{cor:correctness-sufficient}~(p.~\pageref{cor:correctness-sufficient}).
\end{xproof}

\subsection{Order and Maximum Degree of Cluster Tree Graphs}\label{macro:properties}

As already observed,
the CT skeleton $CT_k$ constrains the number of nodes in a CT graph $G_k$ in several ways:

\begin{observation}[Order constraints for $G_k$ from $CT_k$]\label{obs:orderconstraints}
	\hspace*{6cm}
	\begin{enumerate}
		\item For $k>1$, 
		the number of nodes in a cluster of $CT_k$ on level $k+1$ must be at least $\beta^{k}$, since there is always at least one branch instantiation cluster $C'$ on level $k+1$, and nodes in the parent cluster of $C'$ have $\beta^k$ neighbors in $C'$.
		\item To ensure that the edge labels $(\beta^i,\beta^{i+1})$ define feasible biregular bipartite graphs for all $i\in[k]_0$,
		the number of nodes in a single cluster must fall by a factor of $\beta$ per level.
		\item For $k>1$, 
		in the smallest CT graph $G_k$, 
		clusters on level $l$ have $\beta^{2k-l+1}$ nodes, 
		e.g., a cluster on level $k+1$ has $\beta^k$ nodes, and
		$C_0$ has $\beta^{2k+1}$ nodes.
	\end{enumerate}
\end{observation}
\begin{proof}
	Follows immediately from the connectivity structure prescribed by $CT_k$.
\end{proof}

Further, we can determine the number of clusters on each level of $CT_k$:

\begin{theorem}[Order of CT skeletons by level]\label{thm:ct-nclusters}
	For $k\in \mathbb{N}$, the number of clusters $n_C$ on level $l\in \mathbb{N}_0$ in $CT_k$ is\footnote{%
		Following widespread conventions, we set $0!:=1$. 
	}
	\begin{align*}
	n_C(k,l)=\begin{cases}
	1&l=0\\
	\frac{k!}{(k-l+1)!}\cdot (k-l+2) 
	&1\leq l \leq k+1\\
	0&l>k+1.
	\end{cases}
	\end{align*}
\end{theorem}
\begin{proof}
	For $l=1$, $\frac{k!}{(k-l+1)!}\cdot (k-l+2) = k+1$.
	We proceed by induction on $k$. 
	For $k=1$, we have one cluster on the zeroth level, two clusters on the first level, and one cluster on the second level, 
	cf.\ Figure~\ref{fig:ct1} (p.~\pageref{fig:ct1}). 
	Since $n_C(1,0)=1$, $n_C(1,1)=\frac{1!}{1!}\cdot 2=2$,
	$n_C(1,2)=\frac{1!}{0!}\cdot 1=1$, 
	and $n_C(1,l)=0$ for all $l>2$, the claim holds for the base case, 
	i.e., for $k=1$ and all $l$. 
	
	Therefore, assume that the claim holds for some $k$, i.e., in $CT_k$, the number of clusters on level $l$ is given by $n_C(k,l)$. 
	Due to Definition~\ref{def:ct-growth} (p.~\pageref{def:ct-growth}), 
	which enforces that all new leaf clusters lie on the level above their parent clusters, 
	the number of clusters on the zeroth level always remains one.
	This level-$0$ cluster is an internal cluster in $CT_1$ and hence also in $CT_k$. 
	By growth rule~\ref{enum:ctgrowth:rule1}, exactly one cluster is added on the level above the zeroth level when transitioning to $CT_{k+1}$. 
	Therefore, if $CT_k$ has $k+1$ clusters on the first level, $CT_{k+1}$ has $k+2$ clusters on the first level.
	Furthermore, new clusters are added only on levels at most one above already existing clusters, so if $CT_k$ has no clusters on levels above $k+1$, 
	$CT_{k+1}$ has no clusters on levels above $k+2$. 
	With these observations, for $l=0$, $l=1$, and $l>k+1$, the number of clusters on level $l$ in $CT_{k+1}$ is given by the formula stated in the theorem.
	
	For the remaining levels, i.e., levels $l$ with $1< l\leq k+1$, observe that by the growth rules of $CT_k$, 
	all clusters present on level $l$ in $CT_k$ are guaranteed to be internal clusters in $CT_{k+1}$, with $k+1$ children on level $l+1$, 
	and these child clusters are the only clusters on level $l+1$. 
	Thus, for $k, l\geq 1$, the number of clusters satisfies the recurrence relation
	\begin{align*}
	n_C(k+1,l+1)=(k+1) \cdot n_C(k,l).
	\end{align*}
	By the inductive hypothesis, we have that $n_C(k,l)=\frac{k!}{(k-l+1)!}\cdot (k-l+2)$ 
	for $1\leq l\leq k+1$,~so
	\begin{align*}
	n_C(k+1,l+1)
	=\frac{(k+1)! (k-l+2)}{(k-l+1)!}\phantom{,}\\ 
	=\frac{(k+1)!}{((k+1)-(l+1)+1)!}\cdot ((k+1)-(l+1)+2),
	\end{align*}
	as required.
	As this verifies the claimed expression for $n_C(k+1,l)$ for all $1<l\leq k+2$,
	this completes the inductive step, concluding the proof.
\end{proof}

This allows us to express the order of $G_k$ in terms of $n_0:=|C_0|$, $k$, and $\beta$:

\begin{lemma}[$n$ in terms of $n_0$]\label{lem:largec0}\ \\ 
	In $G_k$, $n< n_0\big(\frac{\beta}{\beta-(k+1)}\big)$ and $n-n_0< n_0\cdot\frac{2(k+1)}{\beta}$.
\end{lemma}
\begin{proof}
	By Theorem~\ref{thm:ct-nclusters} (p.~\pageref{thm:ct-nclusters}) and Observation~\ref{obs:orderconstraints} (p.~\pageref{obs:orderconstraints}), 
	the number of nodes on level $l\in\mathbb{N}_0$ in $G_k$ as a function of $\beta$, $k$, $l$, and $n_0$ is 
	$n(\beta,k,l, n_0)=n_C(k,l)\cdot n_0\cdot \beta^{-l}$. 
	Because $n=\sum_{l=0}^{\infty}n(\beta,k,l,n_0)=\sum_{l=0}^{k+1}n(\beta,k,l,n_0)$, we get
	\begin{align*}
	n=n_0+n_0\cdot\sum_{l=1}^{k+1}\frac{k!}{(k-l+1)!}\cdot(k-l+2)\cdot \beta^{-l}
	\phantom{,}\\\leq n_0+n_0\cdot\sum_{l=1}^{k+1}\frac{k!(k+1)}{(k-l+1)!}\cdot\frac{1}{\beta^l}\phantom{,}\\
	< n_0+n_0\cdot\sum_{l=1}^{k+1}\frac{(k+1)^l}{\beta^l}
	= n_0\sum_{l=0}^{k+1}\bigg(\frac{k+1}{\beta}\bigg)^l
	\phantom{,}\\< n_0\sum_{l=0}^{\infty}\bigg(\frac{k+1}{\beta}\bigg)^l = n_0\cdot \frac{\beta}{\beta - (k+1)},
	\end{align*}
	where the last step uses that $\beta>k+1$. 
	Using our requirement that $\beta\geq 2(k+1)$, we obtain
	\begin{align*}
	n-n_0<n_0 \cdot \left(\frac{\beta}{\beta-(k+1)}-1\right)\phantom{,}\\
	=n_0\cdot\frac{k+1}{\beta-(k+1)}\le n_0\cdot\frac{2(k+1)}{\beta},
	\end{align*}
	as desired.
\end{proof}
Finally, the construction of $CT_k$ dictates the largest degree $\Delta$ of a node in $G_k$:

\begin{lemma}[Largest degree $\Delta$ of $G_k$]\label{lem:delta-gk}
	The largest degree of a node in $G_k$ is $\Delta = \beta^{k+1}$.
\end{lemma}
\begin{proof}
	By construction, all nodes in internal clusters have degree $\sum_{i=0}^{k}\beta^i$, and the largest degree of nodes in leaf clusters is $\beta^{k+1}$.
	As $\beta\geq 2(k+1) > 2$, $\max\{\sum_{i=0}^{k}\beta^i, \beta^{k+1}\}=\beta^{k+1}$.
\end{proof}

%% file: tikz-ct2-flat-clustertypes.tex
\begin{tikzpicture}[shorten >=1pt,node distance=2cm and 2cm,on grid,auto] 
\node[state] (c00) [align=center,line width=0.5mm,fill=gray!30,rectangle] {$C_0$};
\node[state] (c01) [align=center,fill=gray!30,line width=0.5mm,rectangle, right=of c00] {$C_1$};
\node[state] (c02) [align=center,fill=gray!30, rectangle, left=of c00] {$C_2$};
\node[state] (c03) [align=center,fill=gray!30, rectangle, right=of c01] {$C_3$};
\node[state] (c11) [align=center,fill=black, below=of c00] {};
\node[state] (c12) [align=center,fill=black, below=of c01] {};
\node[state] (c22) [align=center,fill=green!30!darkgray, left=of c02] {};
\node[state] (c21) [align=center,fill=green!30!darkgray, below=of c02] {};
\node[state] (c32) [align=center,fill=green!30!darkgray, below=of c03] {};
\node[state] (c31) [align=center,fill=green!30!darkgray, right=of c03] {};
\path
	(c00) edge [pos=0.25, align=left,below] node {$0$} (c01)
	(c01) edge [pos=0.25, align=left,below] node {$1$} (c00)
	(c00) edge [pos=0.25, align=left,below] node {$1$} (c02)
	(c02) edge [pos=0.25, align=left,below] node {$2$} (c00)
	(c01) edge [pos=0.25, align=left,below] node {$0$} (c03)
	(c03) edge [pos=0.25, align=left,below] node {$1$} (c01)
	(c02) edge [pos=0.25, align=left,below] node {$1$} (c22)
	(c22) edge [pos=0.25, align=left,below] node {$2$} (c02)
	(c03) edge [pos=0.25, align=left,below] node {$0$} (c31)
	(c31) edge [pos=0.25, align=left,below] node {$1$} (c03)
	(c02) edge [pos=0.25, align=left,left] node {$0$} (c21)
	(c21) edge [pos=0.25, align=left,left] node {$1$} (c02)
	(c02) edge [pos=0.25, align=left,left] node {$0$} (c21)
	(c21) edge [pos=0.25, align=left,left] node {$1$} (c02)
	(c00) edge [pos=0.25, align=left,left] node {$2$} (c11)
	(c11) edge [pos=0.25, align=left,left] node {$3$} (c00)
	(c01) edge [pos=0.25, align=left,left] node {$2$} (c12)
	(c12) edge [pos=0.25, align=left,left] node {$3$} (c01)
	(c03) edge [pos=0.25, align=left,left] node {$2$} (c32)
	(c32) edge [pos=0.25, align=left,left] node {$3$} (c03)
	;
\end{tikzpicture}

%% file: tikz-ct3-flat-clustertypes.tex
\begin{tikzpicture}[shorten >=0pt,node distance=2cm and 2cm,on grid,auto] 
\node[state] (c00) [align=center,fill=gray!30,line width=0.5mm,rectangle] {$C_0$};
\node[state] (c01) [align=center,fill=gray!30,line width=0.5mm,rectangle, right=of c00] {$C_1$};
\node[state] (c02) [align=center,fill=gray!30, rectangle, left=of c00] {$C_2$};
\node[state] (c03) [align=center,fill=gray!30, rectangle, right=of c01] {$C_3$};
\node[state] (c11) [align=center,fill=gray!30, rectangle, below=of c00] {};
\node[state] (c12) [align=center,fill=gray!30, rectangle, below=of c01] {};
\node[state] (c22) [align=center,fill=gray!30, rectangle, left=of c02] {};
\node[state] (c21) [align=center,fill=gray!30, rectangle, below=of c02] {};
\node[state] (c32) [align=center,fill=gray!30, rectangle, below=of c03] {};
\node[state] (c31) [align=center,fill=gray!30, rectangle, right=of c03] {};
\node[state] (c41) [align=center,fill=green!30!darkgray, right=0.74cm of c31,scale=0.5] {};
\node[state] (c42) [align=center,fill=green!30!darkgray, above=0.74cm of c31,scale=0.5] {};
\node[state] (c43) [align=center,fill=green!30!darkgray, below=0.74cm of c31,scale=0.5] {};
\node[state] (c44) [align=center,fill=green!30!darkgray, left=0.74cm of c22,scale=0.5] {};
\node[state] (c45) [align=center,fill=green!30!darkgray, above=0.74cm of c22,scale=0.5] {};
\node[state] (c46) [align=center,fill=green!30!darkgray, below=0.74cm of c22,scale=0.5] {};
\node[state] (c87) [align=center,fill=green!30!darkgray, left=0.74cm of c32,scale=0.5] {};
\node[state] (c88) [align=center,fill=green!30!darkgray, right=0.74cm of c32,scale=0.5] {};
\node[state] (c89) [align=center,fill=green!30!darkgray, below=0.74cm of c32,scale=0.5] {};
\node[state] (c57) [align=center,fill=green!30!darkgray, left=0.74cm of c21,scale=0.5] {};
\node[state] (c58) [align=center,fill=green!30!darkgray, right=0.74cm of c21,scale=0.5] {};
\node[state] (c59) [align=center,fill=green!30!darkgray, below=0.74cm of c21,scale=0.5] {};
\node[state] (c67) [align=center,fill=green!30!darkgray, left=0.74cm of c11,scale=0.5] {};
\node[state] (c68) [align=center,fill=green!30!darkgray, right=0.74cm of c11,scale=0.5] {};
\node[state] (c69) [align=center,fill=green!30!darkgray, below=0.74cm of c11,scale=0.5] {};
\node[state] (c77) [align=center,fill=green!30!darkgray, left=0.74cm of c12,scale=0.5] {};
\node[state] (c78) [align=center,fill=green!30!darkgray, right=0.74cm of c12,scale=0.5] {};
\node[state] (c79) [align=center,fill=green!30!darkgray, below=0.74cm of c12,scale=0.5] {};
\node[state] (c47) [align=center,fill=black, above=0.74cm of c00,scale=0.5] {};
\node[state] (c48) [align=center,fill=black, above=0.74cm of c01,scale=0.5] {};
\node[state] (c49) [align=center,fill=black, above=0.74cm of c02,scale=0.5] {};
\node[state] (c50) [align=center,fill=black, above=0.74cm of c03,scale=0.5] {};
\path
	(c00) edge [pos=0.25, align=left,below] node {$0$} (c01)
	(c01) edge [pos=0.25, align=left,below] node {$1$} (c00)
	(c00) edge [pos=0.25, align=left,below] node {$1$} (c02)
	(c02) edge [pos=0.25, align=left,below] node {$2$} (c00)
	(c01) edge [pos=0.25, align=left,below] node {$0$} (c03)
	(c03) edge [pos=0.25, align=left,below] node {$1$} (c01)
	(c02) edge [pos=0.25, align=left,below] node {$1$} (c22)
	(c22) edge [pos=0.25, align=left,below] node {$2$} (c02)
	(c03) edge [pos=0.25, align=left,below] node {$0$} (c31)
	(c31) edge [pos=0.25, align=left,below] node {$1$} (c03)
	(c02) edge [pos=0.25, align=left,left] node {$0$} (c21)
	(c21) edge [pos=0.25, align=left,left] node {$1$} (c02)
	(c02) edge [pos=0.25, align=left,left] node {$0$} (c21)
	(c21) edge [pos=0.25, align=left,left] node {$1$} (c02)
	(c00) edge [pos=0.25, align=left,left] node {$2$} (c11)
	(c11) edge [pos=0.25, align=left,left] node {$3$} (c00)
	(c01) edge [pos=0.25, align=left,left] node {$2$} (c12)
	(c12) edge [pos=0.25, align=left,left] node {$3$} (c01)
	(c03) edge [pos=0.25, align=left,left] node {$2$} (c32)
	(c32) edge [pos=0.25, align=left,left] node {$3$} (c03)
	(c47) edge [] node {} (c00)
	(c48) edge [] node {} (c01)
	(c49) edge [] node {} (c02)
	(c50) edge [] node {} (c03)
	(c41) edge [] node {} (c31)
	(c42) edge [] node {} (c31)
	(c43) edge [] node {} (c31)
	(c44) edge [] node {} (c22)
	(c45) edge [] node {} (c22)
	(c46) edge [] node {} (c22)
	(c87) edge [] node {} (c32)
	(c88) edge [] node {} (c32)
	(c89) edge [] node {} (c32)
	(c57) edge [] node {} (c21)
	(c58) edge [] node {} (c21)
	(c59) edge [] node {} (c21)
	(c67) edge [] node {} (c11)
	(c68) edge [] node {} (c11)
	(c69) edge [] node {} (c11)
	(c77) edge [] node {} (c12)
	(c78) edge [] node {} (c12)
	(c79) edge [] node {} (c12)
	;
\end{tikzpicture}

%% file: algorithm-indistuinguishability.tex
\begin{algorithm2e}[p!]
	\DontPrintSemicolon 
	\SetKwFunction{FMain}{{\sc FindIsomorphism}}
	\SetKwProg{Fn}{Function}{:}{}
	\Fn{\FMain{$G_k$, $k$, $v_0$, $v_1$}}{
		\KwIn{A CT graph $G_k$ with $g\geq 2k+1$, $k\in \mathbb{N}$, $v_0\in C_0$, $v_1\in C_1$}
		\KwOut{Isomorphism $\phi: V(G_k^k(v_0))\rightarrow V(G_k^k(v_1))$}
		$\phi \gets$ empty map\;
		$\phi(v_0) \gets v_1$\;
		\textsc{Walk($v_0$, $v_1$, $\bot$, $k$)}\;
		\Return{$\phi$}\;
	}
	\BlankLine
	\SetKwFunction{FWalk}{{\sc Walk}}
	\SetKwProg{Pn}{Function}{:}{{\sc Walk}}
	\Pn{\FWalk{$v$, $w$, $prev$, $depth$}}{
		\If{$depth = 0$}{
			\KwRet\;
		}
		$N_v\gets$ empty list of length $k+2$\;\label{alg:iso:neighborhood-start}
		$N_w\gets$ empty list of length $k+2$\;
		\For{$i \gets 0$ \emph{\textbf{to}} $k+1$} {
			\tcp{if edge $\beta^i$ does not exist, $N_v[i]$ (resp.\ $N_w[i]$) is empty}
			$N_v[i]\gets$ list of new nodes $v'\neq prev$ found using edge $\beta^i$ from $v$\;
			$N_w[i]\gets$ list of new nodes $w'\neq \phi(prev)$ found using edge $\beta^i$ from $w$\;\label{alg:iso:neighborhood-end}
		}
		\textsc{Map($N_v$, $N_w$)}\;\label{alg:iso:call-map}
		\For{$i \gets 0$ \emph{\textbf{to}} $k+1$} {
			\For{$v'$ \emph{\textbf{in}} $N_v[i]$}{
				\textsc{Walk($v'$, $\phi(v')$, $v$, $depth-1$)}
			}
		}
	}
	\BlankLine
	\SetKwFunction{FPair}{{\sc Map}}
	\SetKwProg{FPn}{Function}{:}{{\sc Map}}
	\FPn{\FPair{$N_v$, $N_w$}}{
		\For{$i \gets 0$ \emph{\textbf{to}} $k+1$\label{alg:iso:map-outer-for-start}} {
			\tcp{$zip(\cdot,\cdot)$ yields element tuples until the shorter list ends}
			\For{$v',w'$ \emph{\textbf{in}} $zip(N_v[i], N_w[i])$}{
				$\phi(v')\gets w'$\;\label{alg:iso:map-inner-for-end}
			}
		}
	\tcp{$len(\cdot)$ returns the length of a list}
	\If{$\exists~i\in [k+1]_0:  len(N_v[i]) \neq len(N_w[i])$\label{alg:iso:special-start}}{
	\tcp{we will prove that $len(L_v[i]) = len(L_w[i])$ for $i\in[k+1]_0\setminus \{i_v,i_w\}$}
	$i_v \gets i \in [k+1]_0: len(N_v[i]) = len(N_w[i]) + 1$\;
	$i_w \gets i \in [k+1]_0: len(N_v[i]) + 1 = len(N_w[i])$\;
	\tcp{$L[i][-1]$ retrieves the last element from list $i$ in~$L$}
	$\phi(N_v[i_v][-1]) \gets N_w[i_w][-1]$\;\label{alg:iso:special-end}
}
	
	}
	\caption{Find an isomorphism\newline $\phi: V(G_k^k(v_0))\rightarrow V(G_k^k(v_1))$}\label{alg:isomorphism}
\end{algorithm2e}

%% file: tikz-ct2-algorithm-3.tex
\begin{tikzpicture}[shorten >=0pt,node distance=2cm and 2cm,on grid,auto] 
\node[state] (c00) [align=center,line width=0.5mm,fill=gray!30,rectangle] {$C_0$};
\node[state] (v0) [scale=0.5,align=center,line width=0.0mm,fill=red,above right=0.5cm of c00] {\huge $v_0$};
\node[state] (c01) [align=center,fill=gray!30,line width=0.5mm,rectangle, right=of c00] {$C_1$};
\node[state] (v1) [scale=0.5,align=center,line width=0.0mm,fill=cyan!30,above right=0.5cm of c01] {\huge $v_1$};
\node[state] (c02) [align=center,fill=gray!30, rectangle, left=of c00] {$C_2$};
\node[state] (c03) [align=center,fill=gray!30, rectangle, right=of c01] {$C_3$};
\node[state] (c11) [align=center,fill=black, below=of c00] {};
\node[state] (c12) [align=center,fill=black, below=of c01] {};
\node[state] (c22) [align=center,fill=green!30!darkgray, left=of c02] {};
\node[state] (c21) [align=center,fill=green!30!darkgray, below=of c02] {};
\node[state] (c32) [align=center,fill=green!30!darkgray, below=of c03] {};
\node[state] (c31) [align=center,fill=green!30!darkgray, right=of c03] {};
\node[state] (v00) [scale=0.25,align=center,line width=0.0mm,fill=blue,above left=0.5cm of c01] {};
\node[state] (v01) [scale=0.25,align=center,line width=0.0mm,fill=orange,below right=0.5cm of c02] {};
\node[state] (v02) [scale=0.25,align=center,line width=0.0mm,fill=green!30,above right=0.5cm of c11] {};
\node[state] (v10) [scale=0.25,align=center,line width=0.0mm,fill=blue,above left=0.5cm of c03] {};
\node[state] (v11) [scale=0.25,align=center,line width=0.0mm,fill=orange,below right=0.5cm of c00] {};
\node[state] (v12) [scale=0.25,align=center,line width=0.0mm,fill=green!30,above right=0.5cm of c12] {};
\path
	(c00) edge [pos=0.25, align=left,below] node {$0$} (c01)
	(c01) edge [pos=0.25, align=left,below] node {$1$} (c00)
	(c00) edge [pos=0.25, align=left,below] node {$1$} (c02)
	(c02) edge [pos=0.25, align=left,below] node {$2$} (c00)
	(c01) edge [pos=0.25, align=left,below] node {$0$} (c03)
	(c03) edge [pos=0.25, align=left,below] node {$1$} (c01)
	(c02) edge [pos=0.25, align=left,below] node {$1$} (c22)
	(c22) edge [pos=0.25, align=left,below] node {$2$} (c02)
	(c03) edge [pos=0.25, align=left,below] node {$0$} (c31)
	(c31) edge [pos=0.25, align=left,below] node {$1$} (c03)
	(c02) edge [pos=0.25, align=left,left] node {$0$} (c21)
	(c21) edge [pos=0.25, align=left,left] node {$1$} (c02)
	(c02) edge [pos=0.25, align=left,left] node {$0$} (c21)
	(c21) edge [pos=0.25, align=left,left] node {$1$} (c02)
	(c00) edge [pos=0.25, align=left,left] node {$2$} (c11)
	(c11) edge [pos=0.25, align=left,left] node {$3$} (c00)
	(c01) edge [pos=0.25, align=left,left] node {$2$} (c12)
	(c12) edge [pos=0.25, align=left,left] node {$3$} (c01)
	(c03) edge [pos=0.25, align=left,left] node {$2$} (c32)
	(c32) edge [pos=0.25, align=left,left] node {$3$} (c03)
	;
\path[->]
	(v0) edge [color=blue, line width=0.05cm, bend left=30] node {} (v00)
	(v0) edge [color=orange, line width=0.05cm, bend left=30] node {} (v01)
	(v0) edge [color=green!30, line width=0.05cm, bend left=30] node {} (v02)
	(v1) edge [color=blue, line width=0.05cm, bend left=30] node {} (v10)
	(v1) edge [color=orange, line width=0.05cm, bend left=30] node {} (v11)
	(v1) edge [color=green!30, line width=0.05cm, bend left=30] node {} (v12)
	;
\end{tikzpicture}

%% file: tikz-ct2-algorithm-5.tex
\begin{tikzpicture}[shorten >=0pt,node distance=2cm and 2cm,on grid,auto] 
\node[state] (c00) [align=center,line width=0.5mm,fill=gray!30,rectangle] {$C_0$};
\node[state] (v0) [scale=0.5,align=center,line width=0.0mm,fill=red,above right=0.5cm of c00] {\huge $v_0$};
\node[state] (c01) [align=center,fill=gray!30,line width=0.5mm,rectangle, right=of c00] {$C_1$};
\node[state] (v1) [scale=0.5,align=center,line width=0.0mm,fill=cyan!30,above right=0.5cm of c01] {\huge $v_1$};
\node[state] (c02) [align=center,fill=gray!30, rectangle, left=of c00] {$C_2$};
\node[state] (c03) [align=center,fill=gray!30, rectangle, right=of c01] {$C_3$};
\node[state] (c11) [align=center,fill=black, below=of c00] {};
\node[state] (c12) [align=center,fill=black, below=of c01] {};
\node[state] (c22) [align=center,fill=green!30!darkgray, left=of c02] {};
\node[state] (c21) [align=center,fill=green!30!darkgray, below=of c02] {};
\node[state] (c32) [align=center,fill=green!30!darkgray, below=of c03] {};
\node[state] (c31) [align=center,fill=green!30!darkgray, right=of c03] {};
\node[state] (v01) [scale=0.25,align=center,line width=0.0mm,fill=orange,below right=0.5cm of c02] {};
\node[state] (v11) [scale=0.25,align=center,line width=0.0mm,fill=orange,below right=0.5cm of c00] {};
%
\node[state] (v010) [scale=0.25,align=center,line width=0.0mm,fill=blue,above left=0.5cm of c21] {};
\node[state] (v110) [scale=0.25,align=center,line width=0.0mm,fill=blue,above left=0.5cm of c01] {};
\node[state] (v011) [scale=0.25,align=center,line width=0.0mm,fill=orange,above right=0.5cm of c22] {};
\node[state] (v111) [scale=0.25,align=center,line width=0.0mm,fill=orange,above right=0.5cm of c02] {};
\node[state] (v012) [scale=0.25,align=center,line width=0.0mm,fill=green!30,below left=0.5cm of c00] {};
\node[state] (v112) [scale=0.25,align=center,line width=0.0mm,fill=green!30,above left=0.5cm of c11] {};
\path
	(c00) edge [pos=0.25, align=left,below] node {$0$} (c01)
	(c01) edge [pos=0.25, align=left,below] node {$1$} (c00)
	(c00) edge [pos=0.25, align=left,below] node {$1$} (c02)
	(c02) edge [pos=0.25, align=left,below] node {$2$} (c00)
	(c01) edge [pos=0.25, align=left,below] node {$0$} (c03)
	(c03) edge [pos=0.25, align=left,below] node {$1$} (c01)
	(c02) edge [pos=0.25, align=left,below] node {$1$} (c22)
	(c22) edge [pos=0.25, align=left,below] node {$2$} (c02)
	(c03) edge [pos=0.25, align=left,below] node {$0$} (c31)
	(c31) edge [pos=0.25, align=left,below] node {$1$} (c03)
	(c02) edge [pos=0.25, align=left,left] node {$0$} (c21)
	(c21) edge [pos=0.25, align=left,left] node {$1$} (c02)
	(c02) edge [pos=0.25, align=left,left] node {$0$} (c21)
	(c21) edge [pos=0.25, align=left,left] node {$1$} (c02)
	(c00) edge [pos=0.25, align=left,left] node {$2$} (c11)
	(c11) edge [pos=0.25, align=left,left] node {$3$} (c00)
	(c01) edge [pos=0.25, align=left,left] node {$2$} (c12)
	(c12) edge [pos=0.25, align=left,left] node {$3$} (c01)
	(c03) edge [pos=0.25, align=left,left] node {$2$} (c32)
	(c32) edge [pos=0.25, align=left,left] node {$3$} (c03)
	;
\path[->]
	(v0) edge [color=orange, line width=0.05cm, bend left=30] node {} (v01)
	(v1) edge [color=orange, line width=0.05cm, bend left=30] node {} (v11)
	(v01) edge [color=blue, line width=0.05cm, bend right=30] node {} (v010)
	(v11) edge [color=blue, line width=0.05cm, bend right=30, dashed] node {} (v110)
	(v01) edge [color=orange, line width=0.05cm, bend right=30] node {} (v011)
	(v11) edge [color=orange, line width=0.05cm, bend right=30] node {} (v111)
	(v01) edge [color=green!30, line width=0.05cm, bend right=30, dashed] node {} (v012)
	(v11) edge [color=green!30, line width=0.05cm, bend right=30] node {} (v112)
	;
\end{tikzpicture}

%% file: micro_level.tex
\section{Ensuring High Girth}\label{sec:micro_level}

To construct $G_k$ with high girth, we rely on special \emph{graph homomorphisms} called \emph{graph lifts}:

\begin{definition}[Graph homomorphism]\label{def:graph-homomorphism}
	Graph $G_1$ is \emph{homomorphic} to graph $G_2$ if there is a function $\phi\colon V(G_1)\rightarrow V(G_2)$ s.t.\ 
	$\{v,w\}\in E(G_1) \Rightarrow \{\phi(v), \phi(w)\} \in E(G_2)$ (i.e., $\phi$ is adjacency-preserving);
	$\phi$ is called a \emph{homomorphism}.
\end{definition} 

\begin{definition}[Graph lift]\label{def:graph-lift}
	Graph $G_1$ is a \emph{lift} of graph $G_2$ if there is a \emph{surjective} homomorphism $\phi: V(G_1)\rightarrow V(G_2)$~s.t.\ 
	$\forall v \in V(G_1)\colon
	\{v,w\}\in E(G_1) \Leftrightarrow \{\phi\vert_{\Gamma(v)}(v), \phi\vert_{\Gamma(v)}(w)\} \in E(G_2)$ (i.e., $\phi$ is \emph{locally} bijective);
	$\phi$ is called a \emph{covering map}.
\end{definition}

\begin{figure*}[t]
	\centering
	\input{tikz-liftsetup-base}
	\caption[Existence of CT Graphs with High Girth: Setup]{Setup used to establish the existence of CT graphs with high girth.}\label{fig:ct-graphs-highgirth-existence}
\end{figure*}
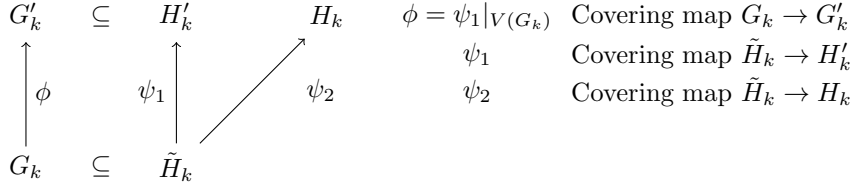
\begin{table*}[t]
	\small
	\centering
	\begin{tabular}{p{0.06\linewidth}p{0.6\linewidth}p{0.25\linewidth}}\hline
		\textbf{Graph}&\textbf{Properties}&\textbf{Existence Proof}\\[3pt]\hline
		$G'_k$&CT graph, parametrized by $\beta$, $\beta^{\mathcal{O}(k)}$ nodes, $\Delta=\beta^{k+1}$&Definition~\ref{def:ct-lowgirth}, Lemma~\ref{lem:ct-lowgirth-structure}\\[3pt]
		$H'_k$&$\Delta$-regular, supergraph of $G'_k$, $\beta^{\mathcal{O}(k)}$ nodes&Lemma~\ref{lem:Hkprime}, Corollary~\ref{cor:Hk}\\[3pt]
		$H_k$&$\Delta$-regular, girth $g=2k+1$, $\mathcal{O}(\Delta^{2k})\subseteq \beta^{\mathcal{O}(k^2)}$ nodes&Lemma~\ref{lem:Hk}, Corollary~\ref{cor:Hk}\\[3pt]
		$\tilde{H}_k$&$\Delta$-regular, lift of $H'_k$ and $H_k$, girth $g\geq 2k+1$, $\beta^{\mathcal{O}(k^2)}$ nodes&Lemma~\ref{lem:Hktilde}, Corollary~\ref{cor:Hktilde}\\[3pt]
		$G_k$&CT graph, subgraph of $\tilde{H}_k$, girth $g\geq 2k+1$, $\beta^{\mathcal{O}(k^2)}$ nodes&Theorem~\ref{thm:Gk-existential}\\[3pt]\hline
	\end{tabular}
	\caption[Existence of CT Graphs with High Girth: Proof Overview]{%
		Proof overview for establishing the existence of CT graphs with high girth.}\label{tab:ct-graphs-highgirth-existence}
\end{table*}

As sketched by Kuhn et al.\ \cite{kuhn2016}, 
we establish the existence of a CT graph $G_k$ with girth $g\geq 2k+1$ and $\mathcal{O}(\beta^{2k^2+4k+1})$ nodes 
using the setup illustrated in Figure~\ref{fig:ct-graphs-highgirth-existence} (p.~\pageref{fig:ct-graphs-highgirth-existence}). 
The intuition of this setup is that we obtain $G_k$ as a subgraph of $\tilde{H}_k$, which is a common lift of a high-girth graph $H_k$ and a graph $H'_k$ that is a supergraph of a low-girth CT graph $G'_k$.
Since taking lifts and subgraphs cannot decrease the girth,\footnote{%
	For subgraphs, this is obvious. 
	For lifts, consider a cycle of length $l$ in the lift and observe that a covering map must map it to a subgraph containing a cycle of length at most $l$.
}
$G_k$ then has large girth,
and because $G_k$ is a lift of a CT graph $G'_k$,
it is a CT graph as well.
Table~\ref{tab:ct-graphs-highgirth-existence} (p.~\pageref{tab:ct-graphs-highgirth-existence}) gives an overview of the graphs involved in our setup, 
along with the properties we seek to establish. 

\subsection{Low-Girth Cluster Tree Graphs}\label{micro:low-girth}

We can easily design low-girth CT graphs by plugging together complete bipartite graphs.

\begin{definition}[$G'_k$ from complete bipartite graphs]\label{def:ct-lowgirth}
	For $k\in\mathbb{N}$ and a parameter $\beta\in\mathbb{N}$, let $CT_k$ be the CT skeleton parametrized by $\beta$. 
	We construct a CT graph $G'_k$ as follows:
	\begin{enumerate}
		\item For cluster $C$ on level $l\in [k+1]_0$ in $CT_k$, add $\beta^{2k-l+1}$ nodes $v$ with $C(v)=C$ to $G_k$.
		\item For clusters $C$ and $C'$ with $\{(C,\beta^x),(C',\beta^{x+1})\} \in E(CT_k)$, 
		connect the nodes representing these clusters in $G_k$ using $\frac{|C|}{\beta^{x+1}}$ copies of $K_{\beta^{x},\beta^{x+1}}$,
		the complete bipartite graph on $A\dot{\cup} B$ with $|A|=\beta^x$ and $|B|=\beta^{x+1}$.
	\end{enumerate}
\end{definition}

\begin{lemma}[CT graph structure]\label{lem:ct-lowgirth-structure}
	Graphs following Definition~\ref{def:ct-lowgirth} (p.~\pageref{def:ct-lowgirth}) are CT graphs.
\end{lemma}
\begin{proof}
	As can be easily verified from the definition of $CT_1$ and the growth rules,
	in $CT_k$, all edges are of the form $\{(C,\beta^x),(C',\beta^{x+1})\}$ for some $x\in [k+1]_0$,
	where $C$ is on level $\ell\in [k]_0$ and $C'$ on level $\ell+1$.
	By Definition~\ref{def:ct-lowgirth} (p.~\pageref{def:ct-lowgirth}),
	we thus have that $|C'|=\frac{|C|}{\beta}$.
	Now, $\frac{|C|}{\beta^{x+1}}$ copies of $K_{\beta^{x},\beta^{x+1}}$ contain $\frac{|C|}{\beta^{x+1}}\cdot \beta^{x+1}=|C|$ nodes with degree $\beta^x$ 
	and $\frac{|C|}{\beta^{x+1}}\cdot \beta^{x}=\frac{|C|}{\beta}=|C'|$ nodes with degree $\beta^{x+1}$. 
	Hence, using $\frac{|C|}{\beta^{x+1}}$ copies of $K_{\beta^{x},\beta^{x+1}}$, 
	we exactly fulfill the requirements imposed by $CT_k$ on the connectivity between $C$ and $C'$.
\end{proof}

\subsection{Regular Graphs with Desirable Properties}\label{micro:regular}

The construction from Definition~\ref{def:ct-lowgirth} (p.~\pageref{def:ct-lowgirth}) results in CT graphs of girth four.
It remains 
to lift these low-girth graphs to high girth.
First, we embed $G'_k$ into a $\Delta$-regular graph.

\begin{lemma}[$\Delta$-regular supergraphs of general graphs]\label{lem:Hkprime}	
	Let $G$ be a simple graph with maximum degree $\Delta$.
	Then there exists a $\Delta$-regular supergraph $H$ of $G$ with $|V(H)|< |V(G)|+4\Delta$.\footnote{%
		A bound of $|V(H)|\leq |V(G)|+\Delta+2$ that is optimal in the worst case is shown in \cite{akiyama1983}.
	}
\end{lemma}
\begin{proof}
	Let $G=(V,E)$ with maximum degree $\Delta$.
	We modify $G$ to form $H$ as follows. 
	While there are nodes $v, w \in V$ with degree $\delta(v)<\Delta$, $\delta(w)< \Delta$, and $\{v,w\}\notin E$, 
	we add $\{v,w\}$ to $E$. 
	Let $D$ be the set of remaining nodes with degree less than $\Delta$. 
	By construction, we know that the nodes in $D$ form a clique of size at most $\Delta$.

	Now add a complete bipartite graph $K_{\Delta,\Delta}$ with node bipartition $\{l_i \mid i \in [\Delta]\}\dot{\cup}\{r_i \mid i \in [\Delta]\}$ 
	and define $\Delta$ disjoint perfect matchings $M_i := \{\{l_x,r_y\} \mid (x-y)\bmod \Delta = i\}$. 
	Assign to each $v\in D$ a unique such matching,
	remove the edges containing $l_i$ for $i\in[\lfloor(\Delta-\delta(v))/2\rfloor]$,
	and connect all endpoints of these edges to $v$.
	Afterwards, nodes in $D$ are missing at most one edge,
	while all other nodes have degree $\Delta$.
	Next, arbitrarily match the nodes still missing edges.
	For each such pair $(v,w)\in D^2$,
	choose the remaining edge from the matching of $v$ that contains $l_{\Delta}$, 
	remove it, connect $w$ to $l_{\Delta}$, and $v$ to the other endpoint.

	After this step, at most one node does not have degree $\Delta$ yet and is missing at most one edge.
	If this case occurs,
	$\Delta$ must be odd (otherwise, $v$ would be the only node with odd degree, while the number of nodes with odd degree in any graph must be even).
	We complete the procedure by adding a copy of $K_{\Delta,\Delta-1}$,
	connecting $v$ to one of the $\Delta$ nodes of degree $\Delta-1$ in $K_{\Delta,\Delta-1}$,
	and adding a perfect matching between the remaining nodes of degree $\Delta-1$ (whose number is even).
	The resulting graph $H$ is a $\Delta$-regular supergraph of $G$ with $|V(H)|\leq |V(G)|+4\Delta -1$ nodes.
\end{proof}

Next, we ensure that $\Delta$-regular graphs of girth $2k+1$ without too many nodes exist:

\begin{lemma}[$\Delta$-regular graphs with prescribed girth and order \cite{erdos1963}]\label{lem:Hk}
	For $2 \leq \Delta \in\mathbb{N}$ and $3\leq g \in \mathbb{N}$, 
	there exist $\Delta$-regular graphs with girth at least $g$ and $2m$ nodes for each 
	$m \geq 2\cdot \sum_{i=0}^{g-2}(\Delta-1)^i$.
\end{lemma}
\begin{proof}
	Fix $g$.
	The claim trivially holds for $\Delta=2$,
	as any cycle of length $2m$ satisfies the requirements.
	Now assume that the claim holds for some $\Delta\geq 2$. 
	Thus, for any $m\geq 2\cdot\sum_{i=0}^{g-2}\Delta^i$, 
	there exists a $\Delta$-regular graph $G$ with $2m\geq 4\cdot \sum_{i=0}^{g-2}(\Delta-1)^i$ nodes and girth at least $g\geq 3$. 
	Now let $G'$ be a graph satisfying the following conditions:
	\begin{enumerate}[noitemsep]
		\item $|V(G')| = 2m$,
		\item $\Delta\leq \delta(v)\leq \Delta+1$ for all nodes $v \in V(G')$,
		\item $G'$ has girth at least $g$, and
		\item $|E(G')|$ is maximal among all graphs (including $G$) that satisfy the other three conditions.
	\end{enumerate} 
	We show that $G'$ is $(\Delta+1)$-regular.
	To this end, assume towards a contradiction 
	that $G'$ is \emph{not} $(\Delta+1)$-regular.
	Then either $G'$ has exactly one node with degree $\Delta$ or $G'$ has at least two nodes $v'$ and $w'$ with degree $\Delta$.

	The first case cannot occur because it would require $G'$ to have exactly one node of odd degree for $\Delta+1$ even, 
	and exactly $2m-1$ nodes of odd degree for $\Delta+1$ odd, 
	contradicting the fact that in any graph, the number of nodes with odd degree must be even.
	So assume that there are at least two nodes $v'$ and $w'$ with degree $\Delta$ in $G'$. 
	Observe that all nodes of degree $\Delta$ must lie within distance $g-2$ of $v'$ and $w'$,
	i.e., in $N := \Gamma^{g-2}(v')\cap \Gamma^{g-2}(w')$,
	as otherwise we could add an edge to $G'$ without violating the first three properties, 
	contradicting the fourth property. 
	Since
	$|\Gamma^j(v)| \leq \sum_{i=0}^{j}\Delta^i$
	for any node $v\in V(G')$ with $\delta(v) = \Delta$, 
	we have
	$|N|\leq m$,
	and consequently,
	$|N|\leq |V(G')|-|N|$.
	
	Now let $\{x, y\}$ be an edge between two nodes $x, y\in V(G')\setminus N$.
	We know such an edge must exist,
	because otherwise $(\Delta+1)\cdot(|V(G')|-|N|)$ edges would need to run between nodes in $V(G')\setminus N$ and nodes in $N$, 
	which would force $\delta(v) = \Delta+1$ for all $v\in N$, 
	contradicting the fact that $\delta(v')=\delta(w')=\Delta$.
	But then $\bar{G}'$ with
	$V(\bar{G}') := V(G')$ and 
	$E(\bar{G}') := (E(G') \setminus \{x, y\}) \cup \{\{x, v'\}, \{y, w'\}\}$
	is a graph with more edges than $G'$ that satisfies the first three requirements 
	(in particular, the new edge set does not introduce a cycle of length $< g$ since $x$ and $y$ lie at distance $\geq g-1$ from $v'$ and $w'$), 
	contradicting the maximality of $G'$.
	Therefore, no node with degree $\Delta$ can exist in $G'$, i.e., $G'$ must be $(\Delta+1)$-regular. 
\end{proof}

\begin{corollary}[Existence of $H'_k$ and $H_k$]\label{cor:Hk}
	There exists a $\Delta$-regular supergraph $H'_k$ of $G'_k$ with $\mathcal{O}(|V(G'_k)|)$ nodes,
	and for $\Delta\geq 2$ and $g\geq 3$, 
	there exists a $\Delta$-regular graph $H_k$ with girth $g=2k+1$ and $\mathcal{O}(\Delta^{2k})$ nodes.
\end{corollary}
\begin{proof}
	The existence of $H'_k$ follows from Lemma~\ref{lem:Hkprime} (p.~\pageref{lem:Hkprime}) as a special case.
	The existence of $H_k$ follows from Lemma~\ref{lem:Hk} (p.~\pageref{lem:Hk}) as a special case, noting that 
	$4\cdot \sum_{i=0}^{(2k+1)-2}\Delta^i 
	= 4\cdot\frac{\Delta^{2k}-1}{\Delta-1}
	\leq 4\cdot \Delta^{2k}~\in\mathcal{O}(\Delta^{2k})$.
\end{proof}

\subsection{High-Girth Cluster Tree Graphs}\label{micro:high-girth}

Our final tool allows us to construct small common lifts of regular graphs:

\begin{lemma}[Common lifts of $\Delta$-regular graphs \cite{angluin1981}]\label{lem:Hktilde}
	Let $H$ and $H'$ be two $\Delta$-regular graphs. 
	Then there exists a graph $\tilde{H}$ that is a lift of $H$ and $H'$ s.t. $|V(\tilde{H})|\leq 4|V(H)||V(H')|$.
\end{lemma}
\begin{proof}
	By Hall's Theorem \cite{konig1916,hall1935}, 
	any regular bipartite graph has a perfect matching, 
	and the edge set of a $\Delta$-regular bipartite graph can be partitioned into $\Delta$ perfect matchings.
	For the special case that $H$ and $H'$ are both bipartite, 
	let $M_1,\dots,M_\Delta$ and $M'_1,\dots,M'_\Delta$ be partitions of their respective edge sets into perfect matchings. 
	We define $\tilde{H}$ with
	$V(\tilde{H}) = V(H)\times V(H')~\text{~and~}~
	E(\tilde{H}) = \{\{(v,w),(v',w')\} \mid \exists i \in [\Delta]\colon \{v,v'\}\in M_i \wedge \{w,w'\}\in M'_i\}$.
	$\tilde{H}$ has $|V(H)||V(H')|$ nodes, 
	and for each $(v,w)\in V(\tilde{H})$ and $i\in [\Delta]$, 
	there are unique $\{v,v'\}\in M_i$ and $\{w,w'\}\in M'_i$. 
	Hence, $(v,w)$ has $\Delta$ neighbors, and if these neighbors are $(v_1,w_1),\dots,(v_\Delta,w_\Delta)$, 
	the neighbors of $v$ in $H$ are $v_1,\dots, v_\Delta$, while the neighbors of $w$ in $H'$ are $w_1,\dots, w_\Delta$.
	Therefore, $\tilde{H}$ is a lift of $H$ via
	$\phi_H: V(\tilde{H})\rightarrow V(H)~\text{~with~}~
	\phi_H((v,w)) = v,$
	and a lift of $H'$ via
	$\phi_{H'}: V(\tilde{H})\rightarrow V(H')~\text{~with~}~
	\phi_{H'}((v,w)) = w$.
	
	If a graph is not bipartite, 
	we construct its canonical double cover, 
	i.e., its tensor product with $K_2$,
	to obtain a bipartite regular graph, 
	with which we proceed as described above.
	The canonical double cover has twice as many nodes as the original graph, 
	and we might need it for $H$ and $H'$, 
	so $|V(\tilde{H})|\leq 4|V(H)||V(H')|$.
	Further, if $\chi_H$ is covering map of the canonical double cover $C$ of $H$, 
	then $\tilde{H}$ is a lift of $H$ via $\phi_C\circ\chi_H$;
	analogously, $\tilde{H}$ is a lift of $H'$.
\end{proof}

\begin{corollary}[Existence of $\tilde{H}_k$]\label{cor:Hktilde}
	There exists a $\Delta$-regular common lift $\tilde{H}_k$ of $H_k$ and $H'_k$ with girth $g\geq 2k+1$ and $\mathcal{O}(\beta^{2k^2+4k+1})$ nodes.
\end{corollary}
\begin{proof}
	Recall that we start with $G'_k$ as specified in  Definition~\ref{def:ct-lowgirth} (p.~\pageref{def:ct-lowgirth}),
	parametrized by $\beta\geq 2k+1$.
	By Lemma~\ref{lem:delta-gk} (p.~\pageref{lem:delta-gk}),
	$G'_k$ has maximum degree $\Delta = \beta^{k+1}$,
	and by Lemma~\ref{lem:largec0} (p.~\pageref{lem:largec0}), 
	$G'_k$ has order $n<|C_0|\big(\frac{\beta}{\beta-(k+1)}\big)\le 2|C_0|=2\beta^{2k+1}$.
	We apply Lemma~\ref{lem:Hkprime} (p.\pageref{lem:Hkprime}) to obtain a $\Delta$-regular supergraph $H'_k$ of $G'_k$ with $\mathcal{O}(\beta^{2k+1})$ nodes.
	By Corollary~\ref{cor:Hk} (p.~\pageref{cor:Hk}), 
	there also exists a $\Delta$-regular graph $H_k$ with girth $\geq 2k+1$ and $|V(H_k)|\in \mathcal{O}(\Delta^{2k})=\mathcal{O}(\beta^{2k^2+2k})$.
	Therefore,
	from Lemma~\ref{lem:Hktilde} (p.~\pageref{lem:Hktilde}) along with the observation that lifting cannot decrease girth,
	we can infer the existence of a graph $\tilde{H}_k$ which is a lift of $H_k$ and $H'_k$,
	has girth at least $2k+1$, and satisfies $|V(\tilde{H})|\leq 4|V(H_k)||V(H'_k)| \in \mathcal{O}(\beta^{2k^2+4k+1})$.
\end{proof}

Finally, we construct $G_k$ as a subgraph of $\tilde{H}_k$:

\begin{theorem}[Existence of $G_k$ \cite{kuhn2016}]\label{thm:Gk-existential}
	There exists a CT graph $G_k$ of girth $g\geq 2k+1$ with $\mathcal{O}(\beta^{2k^2+4k+1})$ nodes.
\end{theorem}
\begin{proof}
	By Corollary~\ref{cor:Hktilde} (p.~\pageref{cor:Hktilde}), 
	there exists a graph $\tilde{H}_k$ with girth $g\geq 2k+1$ and $\mathcal{O}(\beta^{2k^2+4k+1})$ nodes that is a common lift of $H_k'$ and $H_k$.
	Now let $\psi_1$ be a covering map from $\tilde{H}_k$ to $H'_k$.
	We construct $G_k$ as a subgraph of $\tilde{H}_k$ with 
	$V(G_k) := \{v \in \tilde{H}_k \mid \psi_1(v) \in G'_k\}$ and 
	$E(G_k) := \{\{v,w\} \mid v, w \in V(\tilde{H}_k) \wedge \{\psi_1(v),\psi_1(w)\}\in E(G'_k)\}$.
	Then $\phi := \psi_1\vert_{V(G_k)}$ is a covering map from $G_k$ to $G'_k$.
	To see that $G_k$ is a CT graph, 
	observe that $\phi$ is indeed a bijection on node neighborhoods, 
	and set $C(v) := C(\phi(v))$ for $v\in V(G_k)$.
	Hence,
	we can conclude that $G_k$ is a CT graph (inherited from $G'_k$) with  $\mathcal{O}(\beta^{2k^2+4k+1})$ nodes and girth $g\geq 2k+1$ (inherited from $\tilde{H}_k$).
\end{proof}

%% file: tikz-liftsetup-base.tex
\begin{tikzpicture}
[shorten >=1pt,node distance=2cm and 4cm,on grid,auto] 
\node[] (Gkprime) [align=center] {$G'_k$};
\node[right of=Gkprime] (Hprime) [align=center] {$H'_k$};
\node[right of=Hprime] (H) [align=center] {$H_k$};
\node[below of=Gkprime] (Gk) [align=center] {$G_k$};
\node[right of=Gk] (Htilde) [align=center] {$\tilde{H}_k$};
\node[right=1cm of Gkprime] (subsetone) [align=center] {$\subseteq$};
\node[right=1cm of Gk] (subsettwo) [align=center] {$\subseteq$};
\node[right=2cm of H] (f) [align=center] {};
\node[right=1.875cm of f] (properties) [align=left] {};
\node[below=0cm of f] (phi) [align=center] {$\phi = \psi_1|_{V(G_k)}$};
\node[right=3.75cm of phi, text width=5cm] (phiprop) [align=left] {Covering map $G_k\rightarrow G'_k$};
\node[below=0.5cm of phi] (psione) [align=center] {$\psi_1$};
\node[below=0.5cm of phiprop, text width=5cm] (psioneprop) [align=left] {Covering map $\tilde{H}_k\rightarrow H'_k$};
\node[below=0.5cm of psione] (psitwo) [align=center] {$\psi_2$};
\node[below=0.5cm of psioneprop, text width=5cm] (psitwoprop) [align=left] {Covering map $\tilde{H}_k\rightarrow H_k$};
\path[->]
(Htilde) edge [align=center,left] node {$\psi_1$} (Hprime)
(Htilde) edge [align=center,right] node {$~~~~~\psi_2$} (H)
(Gk) edge [align=center,right] node {$\phi$} (Gkprime)
;
\end{tikzpicture}

%% file: bounds.tex
\section{Lower Bound on Minimum Vertex Cover Approximation}\label{sec:bounds}

\begin{definition}[Vertex cover]\label{def:vc}
Given a finite, simple graph $G=(V,E)$,
a \emph{vertex cover} is a node subset $S\subseteq V$ meeting all edges,
i.e., for each $\{v,w\}\in E$, $v\in S$ or $w\in S$.
A \emph{Minimum Vertex Cover (MVC)} is a vertex cover of minimum cardinality,
and an $\alpha$-approximate MVC is a vertex cover that is at most factor $\alpha$ larger than an MVC.
\end{definition}

We begin by bounding the size of an MVC of any CT graph $G_k$.
To this end, recall that $n_0$ is the number of nodes in $C_0$,
which we have shown to contain a large fraction of all nodes.

\begin{observation}[Size of an MVC of $G_k$]\label{obs:mvcingk}
	$|MVC| \leq n-n_0$.
\end{observation}
\begin{proof}
As $C_0$ is an independent set, 
$V(G_k)\setminus C_0$ is a vertex cover.
\end{proof}

Due to the indistinguishability of nodes in $C_0$ and $C_1$ in $G_k$  (Theorem~\ref{thm:indistinguishability}, p.~\pageref{thm:indistinguishability}), 
we obtain the following requirement for the behavior of any $k$-round distributed algorithm:

\begin{lemma}[Size of a computed vertex cover]\label{lem:mvcsize}
On a CT graph $G_k$ of girth at least $2k+1$ with uniformly random node identifiers,
in the worst case (in expectation), 
a $k$-round deterministic (randomized) MVC algorithm in the LOCAL model must select at least $\frac{n_0}{2}$ nodes.
\end{lemma}
\begin{proof}
Recall that a $k$-round LOCAL algorithm is a function $f$ mapping $k$-hop subgraphs labeled by inputs, node identifiers,
and, in case of a randomized algorithm, strings of unbiased independent random bits to outputs (cf. Definition~\ref{def:distributed-algorithm}, p.~\pageref{def:distributed-algorithm}).
Noting that the vertex cover task has no inputs,
by Theorem~\ref{thm:indistinguishability} (p.~\pageref{thm:indistinguishability}),
nodes in $C_0$ and $C_1$ are $k$-hop indistinguishable.
Hence, restricting $f$ to the $k$-hop subgraphs of nodes in $C_0\cup C_1$,
we get a function depending only on the node identifiers and random strings observed up to distance $k$ in the (isomorphic) trees that constitute the $k$-hop subgraphs.
Now assign the node identifiers uniformly at random (from the feasible range, drawn without repetition).
By the above observations,
the output of each node $v\in C_0\cup C_1$ then depends on the random labeling of its $k$-hop subgraph only,
which is drawn from the same distribution for each node.
Thus, there is some $p\in [0,1]$ such that for each $v\in C_0\cup C_1$,
the probability that $v$ enters the vertex cover computed by the algorithm equals $p$.
Now consider $v\in C_0$.
By the construction of $G_k$,
there is some edge $\{v,w\}\in E(G_k)$ such that $w\in C_1$.
Because $v$ or $w$ must be in the vertex cover the algorithm computes,
we have that $1=P[v$ or $w$ are in the vertex cover$]\le 2p$.
By linearity of expectation,
we conclude that the expected size of the vertex cover is at least $p|C_0\cup C_1|\geq \frac{n_0}{2}$.
\end{proof}

Choosing $\beta$ appropriately,
we 
arrive at the desired MVC lower bound:

\begin{theorem}[MVC lower bound]\label{thm:bound-mvc}
In the family of graphs with at most $n$ nodes and degrees of at most $\Delta$,
the worst-case (expected) approximation ratio $\alpha$ of a deterministic (randomized) $k$-round MVC algorithm in the LOCAL model satisfies
$\alpha \in \min\big\{n^{\Omega(1/(k^2\log k))},\Delta^{\Omega(1/(k\log k))}\big\}$.
In particular, achieving an (expected) approximation ratio $\alpha \in \log^{\mathcal{O}(1)}\min\{n,\Delta\}$ requires 
$k\in \Omega(\min\{\sqrt{\log n/\log \log n},\log \Delta/\log \log \Delta\})$ communication rounds.
\end{theorem}
\begin{proof}
Given any $\alpha>1$,
fix $\beta:=4(k+1)\alpha$.
By Theorem~\ref{thm:Gk-existential} (p.~\pageref{thm:Gk-existential}),
CT graphs of girth $2k+1$ with $\mathcal{O}(\beta^{2k^2+4k+1})\subseteq 2^{\mathcal{O}(k^2(\log k+\log \alpha))}$ nodes exist,
which by Lemma~\ref{lem:delta-gk} (p.~\pageref{lem:delta-gk}) have maximum degree $\Delta=\beta^{k+1}\in 2^{\mathcal{O}(k(\log k + \log \alpha))}$.
We need these bounds to be smaller than $n$ and $\Delta$, respectively.
As we want to show an asymptotic bound for $\alpha$, we may assume that $n$ and $\Delta$ are sufficiently large constants.
Hence, it is sufficient to satisfy the constraints
$\alpha \le 2^{ck^{-2}\log n - \log k}$ and 
$\alpha \le 2^{ck^{-1}\log \Delta-\log k}$,
respectively, where $c>0$ is a sufficiently small constant.
For $k\le \frac{c}{2}\cdot \min\{\sqrt{{\log n}/{\log \log n}},{\log \Delta}/{\log \log \Delta}\}$,
the $\log k$ terms are dominated and the constraints are met for $\alpha \in \min\{n^{\Omega(1/(k^2\log k))},$ $\Delta^{\Omega(1/(k\log k))}\}$.

In particular, $k \in \omega(\min\{\sqrt{{\log n}/{\log \log n}},$ ${\log \Delta}/{\log \log \Delta}\})$ enables us to choose $\alpha = \min\{\log^{\omega(1)} n,\log^{\omega(1)} \Delta\}$.
Hence, if we can show that for any (feasible) choice of $\alpha$ and CT graph $G_k$ with parameter $\beta = 4(k+1)\alpha$ and girth $2k+1$,
any algorithm in the LOCAL model has approximation ratio at least $\alpha$ in the worst case (in expectation),
the claim of the theorem follows.
To see this,
note that $G_k$ contains a vertex cover of size $n-n_0$ by Observation~\ref{obs:mvcingk} (p.~\pageref{obs:mvcingk}) 
and any $k$-round algorithm selects at least $\frac{n_0}{2}$ nodes (in expectation) under uniformly random node identifiers by Lemma~\ref{lem:mvcsize} (p.~\pageref{lem:mvcsize}).
By Lemma~\ref{lem:largec0} (p.~\pageref{lem:largec0}),
this results in an (expected) approximation ratio of at least
$\frac{n_0}{2(n-n_0)}\ge \frac{\beta}{4(k+1)}=\alpha$.
\end{proof}

%% file: appendix.tex
\section{Further Lower Bounds}\label{app:bounds}
\begin{definition}[Fundamental graph problems]\label{def:graphproblems}
	Given a finite, simple graph $G=(V,E)$,
	
	\begin{itemize}[label=--,noitemsep]
		\item \textbf{Minimum Vertex Cover (MVC)} 
		\dots find a minimum vertex subset $S\subseteq V$ s.t. $\forall~\{u,v\} \in E:~u\in S \vee v\in S$.
		
		\item \textbf{Minimum Dominating Set (MDS)} 
		\dots find a minimum vertex subset $S\subseteq V$ s.t. $\forall~v\in V:~v\in S~\vee~\exists~u\in S: \{u,v\} \in E$.
		
		\item \textbf{Maximum Matching (MaxM)} 
		\dots find a maximum edge subset $T\subseteq E$ s.t. $\forall~e_1, e_2\in T:~e_1\cap e_2 =~\emptyset$.

		\item \textbf{Maximal Matching (MM)} 
		\dots find an inclusion-maximal edge subset $T \subseteq E$ s.t. $\forall~e_1, e_2\in T:~e_1\cap e_2 = \emptyset$.

		\item \textbf{Maximal Independent Set (MIS)} 
		\dots find an inclusion-maximal vertex subset $S\subseteq V$ s.t. $\forall~u, v\in S: \{u,v\}\notin E$.
	\end{itemize}
\end{definition}

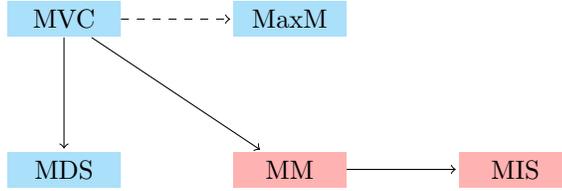
\begin{figure*}[ht!]
	\centering
	\input{tikz-reductions}
	\caption[Relationships between Lower Bounds]{%
		Relationships between the lower bounds derived for fundamental graph problems (adapted from \cite{kuhn2016}). 
		Optimization problems marked blue; 
		binary problems marked red. 
		Solid arrows indicate reductions; 
		dashed arrows indicate analogy.
	}\label{fig:reductions}
\end{figure*}

Figure~\ref{fig:reductions} (p.~\pageref{fig:reductions}) shows how the relationships between these problems are leveraged in the reductions that follow.

\subsection{Minimum Dominating Set (MDS)}\label{bounds:mds}
\begin{theorem}[MDS lower bound]\label{thm:bound-mds}
	The best approximation ratio a $k$-round deterministic (randomized) MDS algorithm in the LOCAL model can achieve is
	$\alpha \in n^{\Omega(1/k^2)}/{k}$ and 
	$\alpha \in \Omega\big({\Delta^{\frac{1}{k+1}}}/{k}\big)$. 
	Hence, to obtain an approximation ratio polylogarithmic in $n$ or $\Delta$, 
	in the worst case (in expectation),
	an algorithm needs to run for 
	$k\in \Omega\big(\sqrt{{\log n}/{\log\log n}}\big)$ or 
	$k\in \Omega\big({\log \Delta}/{\log\log \Delta}\big)$ rounds, respectively. 
\end{theorem}
\begin{proof}
	Observe that a vertex cover $VC$ of any graph $G$ can be transformed into a dominating set $DS$ of its line graph $L(G)$ 
	without increasing its cardinality by adding one edge $\{v,w\}\in V(L(G))$ for every node $v\in V(G)$.
	Similarly, a dominating set of $L(G)$ can be turned into a vertex cover of $G$ by adding the nodes $v, w$ for all $\{v,w\}\in DS(L(G))$, at most doubling its size. 
	Therefore, in general, 
	MVC and MDS are equivalent up to a factor of two in the approximation ratio.
	Now consider a CT graph $G_k$ and its line graph $L(G_k)$. 
	In the LOCAL model of computation, 
	a $k$-round computation on the line graph can be simulated in $k+1$ rounds on the original graph, 
	i.e., $G_k$ and $L(G_k)$ have the same locality properties. 
	Hence, up to a factor of two in the approximation ratio, 
	MVC and MDS are equivalent also in our computational model. 
	The stated bounds hence follow analogously to Theorem~\ref{thm:bound-mvc} (p.~\pageref{thm:bound-mvc}).
\end{proof}

\subsection{Maximum Matching (MaxM)}\label{bounds:maxm}

Here, we are dealing with a packing problem, 
rather than a covering problem,
and we are asked to select edges, rather than nodes.
Therefore, 
instead of using a reduction from MVC, 
we amend the CT graph construction to allow for \emph{edge} indistinguishability arguments:
\begin{definition}[$k$-hop edge indistinguishability]\label{def:indistinguishability-edges}
Two~edges $\{v,w\}$ and $\{v',w'\}$ are $k$-hop indistinguishable if 
there exists an isomorphism $\phi: V(G^k(\{v,w\})) \rightarrow V(G^k(\{v',w'\}))$ with $\phi(v) = v'$ and $\phi(w) = w'$,
where $G^k(\{v,w\}) := G^k(\{v\}) \cup G^k(\{w\})$ and $G^k(\{v',w'\}) := G^k(\{v'\}) \cup G^k(\{w'\})$.
\end{definition}

\begin{theorem}[MaxM lower bound]\label{thm:bound:maxm}
	The best approximation ratio a $k$-round deterministic (randomized) MaxM algorithm in the LOCAL model can achieve is
	$\alpha \in n^{\Omega(1/k^2)}/{k}$ and 
	$\alpha \in \Omega\big({\Delta^{\frac{1}{k+1}}}/{k}\big)$. 
	Hence, to obtain an approximation ratio polylogarithmic in $n$ or $\Delta$, 
	in the worst case (in expectation),
	an algorithm needs to run for 
	$k\in \Omega\big(\sqrt{{\log n}/{\log\log n}}\big)$ or 
	$k\in \Omega\big({\log \Delta}/{\log\log \Delta}\big)$ rounds, respectively. 
\end{theorem}
\begin{proof}
	We create a hard graph $H_k$ from two low-girth copies of $G_k$, $G'_k$ and $\bar{G}'_k$, by first adding a perfect matching
	to connect each node from $G'_k$ with its counterpart in $\bar{G}'_k$ to form a low-girth graph $H'_k$ 
	and then lifting $H'_k$ to high girth using the construction detailed in Section~\ref{sec:micro_level} 
	to obtain $H_k$ with high girth.\footnote{%
		This idea appears already in \cite{kuhn2006},
		but the construction differs from the one presented here in that all powers of $\beta$ are shifted by one, 
		e.g., nodes in $C_0$ have $\beta^1$, rather than $\beta^0$, neighbors in $C_1$.}
	Figure~\ref{fig:hk} (p.~\pageref{fig:hk}) illustrates the idea.
	We refer to the part of $H_k$ corresponding to $G'_k$ as $G_k$ and to the part of $H_k$ corresponding to $\bar{G}'_k$ as $\bar{G}_k$.
	
	Since $G_k$ and $\bar{G}_k$ are high-girth CT graphs, 
	all nodes in the clusters $C_0$, $\bar{C}_0$, $C_1$, 
	and $\bar{C}_1$---and hence, the endpoints of edges $\{v_0,v_1\}$, $\{v_0, \bar{v}_0\}$, $\{\bar{v}_0,\bar{v}_1\}$, and $\{v_1, \bar{v}_1\}$ (where $v_i\in C_i$ and $\bar{v}_i\in \bar{C}_i$ for $i\in\{0,1\}$)---%
	have isomorphic $k$-hop subgraphs if the matching is not added before the lift. 
	In $H_k$, each node from $G_k$ has $1 = \beta^0$ additional neighbor in the copy of its own cluster in $\bar{G}_k$. 
	Since $H_k$ has high girth, however, the $k$-hop subgraphs of nodes in $C_0$, $\bar{C}_0$, $C_1$, 
	and $\bar{C}_1$ are still $k$-hop indistinguishable, 
	and an isomorphism $\phi: V(G^k_k(\{v_0,v_1\})) \rightarrow V(G^k_k(\{v_0,\bar{v}_0\}))$ ($\phi: V(G^k_k(\{\bar{v}_0,\bar{v}_1\})) \rightarrow V(G^k_k(\{v_1,\bar{v}_1\}))$) can map nodes from these clusters onto each other as needed to satisfy the requirements of Definition~\ref{def:indistinguishability-edges} 
	(we could again define an algorithm analogous to Algorithm~\ref{alg:isomorphism}, p.~\pageref{alg:isomorphism}, to construct this isomorphism explicitly).
	It follows that the edges running between $C_0$ and $C_1$ ($\bar{C}_0$ and $\bar{C}_1$) are $k$-hop edge indistinguishable from the edges running between $C_0$ and $\bar{C}_0$ ($C_1$ and $\bar{C}_1$).
	
	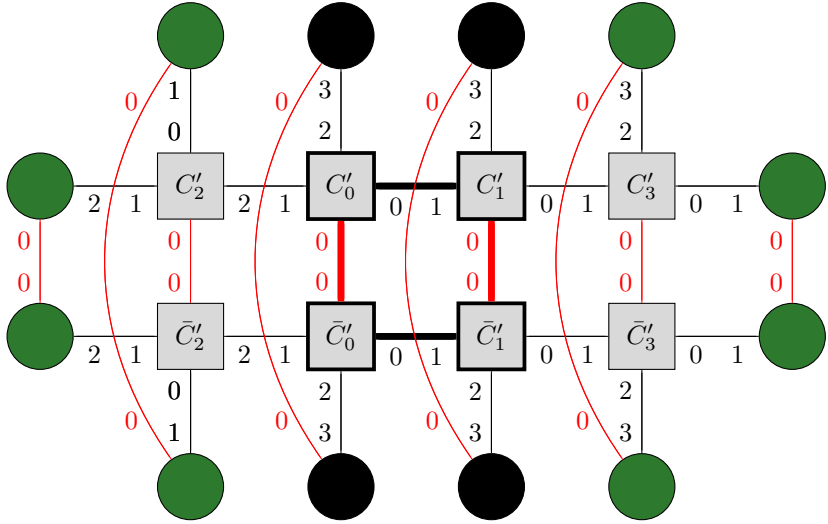
\begin{figure*}[hbt!]
		\centering
		\input{tikz-hk}
		\caption[CT Graphs for Edge Indistinguishability Arguments]{%
			Construction of $H'_2$ from $G'_2$ and $\bar{G}'_2$; 
			edges corresponding to the perfect matching between nodes in $G'_2$ and $\bar{G}'_2$ are marked red. 
			The edges represented by the thickened lines are $2$-hop indistinguishable in $H_2$.
		}\label{fig:hk}
	\end{figure*}
	
	Now consider a node $v\in C_0$ and the set $E_v$ of $\beta + 1$ pairwise indistinguishable edges that have $v$ as an endpoint.
	To guarantee a valid matching, 
	a deterministic algorithm operating on a labeling chosen uniformly at random must ensure $\sum_{e\in E_v}p(e)\leq 1$,
	so each edge $e$ (including the edge running from $v$ to $v'\in \bar{C}_0$) must be selected into the matching with probability $p(e) \leq \frac{1}{\beta +1}$.
	Consequently, 
	the expected number of edges contributed to the matching by edges running between $C_0$ and $\bar{C}_0$
	is 
	$\mathbb{E}[|M(C_0,\bar{C}_0)_D|]\leq \frac{n_0}{\beta+1}$ 
	by linearity of expectation, where $n_0 := |C_0|$.
	To obtain a feasible matching, 
	the number of edges in the matching without an endpoint in $C_0\cup\bar{C}_0$ can be at most $2n-2n_0$, 
	where $n:= |V(G_k)|$.
	It follows that there exists at least one labeling for which a $k$-round deterministic algorithm produces a matching with
	\begin{align*}
	|M_D|\leq \frac{n_0}{\beta+1} + 2n-2n_0&\leq~n_0\cdot \frac{4(k+1)+1}{\beta}\\
	&\in~ \mathcal{O}\left(n\cdot\frac{k}{\beta}\right),
	\end{align*}
	where in the second to last step we applied Lemma~\ref{lem:largec0} (p.~\pageref{lem:largec0}) to bound $n-n_0$.
	With $\mathbb{E}[|M_{R}|]\leq \mathbb{E}[|M_D|]$ from Yao's principle, the bound generalizes to $k$-round randomized algorithms.
	
	To see that this enforces the approximation ratios stated above, 
	observe that the maximum matching for $H_k$ has cardinality $n$ by construction,
	i.e., $\alpha\in \Omega\big(\frac{\beta}{k}\big)$.
	The trade-offs between running time and approximation ratio in terms of $n$ and $\Delta$ now follow analogously to the proof of Theorem~\ref{thm:bound-mvc} (p.~\pageref{thm:bound-mvc}),
	noting that the increase of factor $2$ in the number of nodes and additive $1$ in node degrees has no asymptotic effect.
\end{proof}

\subsection{Maximal Matching (MM)}\label{bounds:mm}

We start by establishing the lower bound in the deterministic setting by exploiting the relationship between MM and MVC:

\begin{theorem}[MM lower bound for deterministic algorithms]\label{thm:bound-mm-deterministic}
	Any \emph{deterministic} MM algorithm needs to run for
	$k~\in~\Omega\big(\min\big\{\sqrt{\log n/\log\log n},$ $\log \Delta / \log\log \Delta \}\big)$
	in the worst case.
\end{theorem}
\begin{proof}
	Since taking the endpoints of a maximal matching yields a $2$-approximation of MVC,
	the claim follows immediately from the bounds established in Theorem~\ref{thm:bound-mvc} (p.~\pageref{thm:bound-mvc}).
\end{proof}

For the randomized setting, we do not obtain the same bounds as in Theorem~\ref{thm:bound-mvc} (p.~\pageref{thm:bound-mvc}) immediately. 
The reason is that randomized algorithms for binary problems lend themselves to Las Vegas algorithms, 
whereas randomized algorithms for optimization problems lend themselves to Monte Carlo algorithms.
We establish the bounds for the randomized setting by showing how a randomized MM algorithm that operates in $T$ rounds in expectation 
can be used to compute an $\mathcal{O}(1)$ approximation in expectation for MVC in $2T+2$ rounds:

\begin{theorem}[MM lower bound for randomized algorithms]\label{thm:bound-mm-randomized}
	In expectation, to find a solution, 
	any \emph{randomized} MM algorithm needs to run for 
	$k\in \Omega\big(\sqrt{{\log n}/{\log\log n}}\big)$ and 
	$k\in \Omega\big({\log \Delta}/{\log\log \Delta}\big)$ rounds. 
\end{theorem}
\begin{proof}
	Let $\mathcal{A}_M$ be an MM algorithm with expected time complexity $T$, running on some graph $G = (V,E)$ with maximum degree $\Delta$.
	The following MVC approximation algorithm $\mathcal{A}_{VC}$ 
	runs with fixed time complexity $2T+2$:
	\begin{enumerate}[noitemsep]
		\item For a sufficiently large constant $c$, execute $c\ln \Delta$ independent runs $i$ of the following in parallel:
		\begin{enumerate}[noitemsep]
			\item All nodes simulate $\mathcal{A}_M$ for $2T$ rounds.
			\item If $E_M$ is the edge set selected after these rounds, every node that is incident with more than one selected edge removes \emph{all} selected incident edges from $E_M$ in an additional round of communication.
			\item All nodes (locally) output the endpoints of all edges remaining in $E_M$ as $V'_{VC, i}$.
		\end{enumerate}
		\item Define 
		$x_v := 6\cdot\frac{|\{i\mid v\in V'_{VC,i}\}|}{c\ln\Delta}~,~\text{and set}~V_{VC} := \{v\in V\mid x_v\geq 1\}$.
		\label{step:first-shot}
		\item All nodes communicate whether they are in $V_{VC}$, 
		and nodes with a remaining uncovered edge join $V_{VC}$.\label{step:last-resort}
	\end{enumerate}
	The final step ensures that the algorithm returns a vertex cover.
	
	To see that not too many nodes are selected in expectation,
	observe first that by construction,
	$V'_{VC,i}$ is a matching for each $i$.
	Therefore, we have that
	$\sum_{v\in V(G)} x_v \le 6\cdot 2 \cdot |MVC|$,
	where $MVC$ is a minimum vertex cover.
	As only nodes with $x_v\ge 1$ are selected in Step~\ref{step:first-shot},
	the total number of nodes selected in this step is (deterministically) at most $12 \cdot |MVC|$.
	
	It remains to bound the expected number of nodes selected in Step~\ref{step:last-resort}.
	To this end,
	observe that by Markov's bound, 
	each independent run of $\mathcal{A}_M$ yields a maximal matching with probability $\geq \frac{1}{2}$,
	and hence, each $V'_{VC, i}$ forms a VC with that same probability.
	Whenever this is the case, $V'_{VC, i}$ contains at least one endpoint of each edge $\{v,w\}\in E$. 
	Hence, if at least one third of all runs are successful, we have $x_v+x_w\geq 2$ for all edges $\{v,w\}\in E$, 
	and $V_{VC}$ is a vertex cover already at the end of Step~\ref{step:first-shot}.
	Letting $X$ be sum of the independent and identically distributed Bernoulli variables $X_i$ indicating whether run $i$ is successful, 
	we have $\mathbb{E}[X] \geq \frac{c\ln \Delta}{2}$.
	Using a Chernoff bound,
	we can then bound the probability to have less than $\frac{c\ln \Delta}{3}$ runs in which $V'_{VC, i}$ forms a VC as
	\begin{align*}
	\mathbb{P}\bigg[X<\frac{c \ln \Delta}{3}\bigg]
	\leq~& 
	\mathbb{P}\bigg[X\leq\left(1-\frac{1}{3}\right)\frac{c \ln \Delta}{2}\bigg]
	\\\leq~ 
	&e^{-\frac{(\frac{1}{3})^2\cdot c\ln \Delta}{4}}
	= e^{-\frac{c\ln \Delta}{36}}
	= \frac{1}{\Delta^{\frac{c}{36}}}.
	\end{align*}
	Hence, for $c\geq 36$, the probability that $V_{VC}$ is not a VC after Step~\ref{step:first-shot} 
	is at most $\frac{1}{\Delta}$. 
	Therefore,
	with probability at least $1-1/\Delta$,
	no further nodes are added in Step~\ref{step:last-resort} of the algorithm.
	Otherwise, i.e.,
	with probability at most $1/\Delta$,
	we add no more than $2|E(G)|$ nodes.
	Given that any vertex cover must contain at least $|E(G)|/\Delta$ nodes,
	we conclude that the expected size of the VC computed via the procedure described above is at most
	\begin{align*}
	12\cdot |MVC| + \frac{1}{\Delta}\cdot 2\Delta\cdot |MVC| = 14\cdot |MVC|.
	\end{align*}
	Thus, a randomized MM algorithm beating the stated bounds would imply an MVC algorithm beating the bounds from Theorem~\ref{thm:bound-mvc} (p.~\pageref{thm:bound-mvc}).
	
	Since such an MVC algorithm cannot exist, 
	the stated bounds must hold.\footnote{%
		In \cite{kuhn2016}, the size of the VC computed in Step~\ref{step:first-shot} is bounded as at most $10\cdot |MVC|$ without explanation, 
		and an expected VC size of at most $11 \cdot |MVC|$ is derived.
	}
\end{proof}

\balance
\subsection{Maximal Independent Set (MIS)}\label{bounds:mis}

We establish our last lower bound via reduction from MM:

\begin{theorem}[MIS lower bound]\label{thm:bound-mis}
	In the worst case (in expectation), to find a solution, 
	any deterministic (randomized) MIS algorithm needs to run for 
	$k\in \Omega\big(\sqrt{{\log n}/{\log\log n}}\big)$ or 
	$k\in \Omega\big({\log \Delta}/{\log\log \Delta}\big)$ rounds. 
\end{theorem}
\begin{proof}
	Observe that an MM of $G_k$ is an MIS of the line graph $L(G_k)$,
	and that a $k$-round MIS computation on $L(G_k)$ can be simulated in $k+1$ rounds on $G_k$.
	Furthermore, $n_{L(G_k)}\leq n_{G_k}^2/2$, 
	and $\Delta_{L(G_k)}\leq 2\Delta_{G_k}$. 
	As $\log n \in \Theta(\log (n^2/2))$, 
	an MIS algorithm beating the stated bounds on $L(G_k)$ would imply an MM algorithm
	beating the bounds from Theorems~\ref{thm:bound-mm-deterministic} or \ref{thm:bound-mm-randomized} (pp.~\pageref{thm:bound-mm-deterministic} et seq.). 
	It follows that the stated bounds must hold also for MIS.
\end{proof}
\vspace*{0pt}

%% file: tikz-reductions.tex
\begin{tikzpicture}
[shorten >=1pt,minimum width=1.5cm,node distance=2cm,on grid,auto] 
\node[fill=cyan!30] (MVC) [align=center] {MVC};
\node[below of=MVC,fill=cyan!30] (MDS) [align=center] {MDS};
\node[right=3cm of MVC,fill=cyan!30] (MaxM) [align=center] {MaxM};
\node[below=of MaxM,fill=red!30] (MM) [align=center] {MM};
\node[right=3cm of MM,fill=red!30] (MIS) [align=center] {MIS};

\path[->]
(MVC) edge [align=center,left] node {} (MDS)
(MVC) edge [align=center,right,dashed] node {} (MaxM)
(MVC) edge [align=center,right] node {} (MM)
(MM) edge [align=center,right] node {} (MIS)
;
\end{tikzpicture}

%% file: tikz-hk.tex
\begin{tikzpicture}[shorten >=1pt,node distance=2cm and 2cm,on grid,auto] 
\node[state] (c00) [align=center,line width=0.5mm,fill=gray!30,rectangle] {$\bar{C}'_0$};
\node[state] (c01) [align=center,fill=gray!30,line width=0.5mm,rectangle, right=of c00] {$\bar{C}'_1$};
\node[state] (c02) [align=center,fill=gray!30, rectangle, left=of c00] {$\bar{C}'_2$};
\node[state] (c03) [align=center,fill=gray!30, rectangle, right=of c01] {$\bar{C}'_3$};
\node[state] (c11) [align=center,fill=black, below=of c00] {};
\node[state] (c12) [align=center,fill=black, below=of c01] {};
\node[state] (c22) [align=center,fill=green!30!darkgray, left=of c02] {};
\node[state] (c21) [align=center,fill=green!30!darkgray, below=of c02] {};
\node[state] (c32) [align=center,fill=green!30!darkgray, below=of c03] {};
\node[state] (c31) [align=center,fill=green!30!darkgray, right=of c03] {};
\path
	(c00) edge [pos=0.25, align=left,below,line width=2.5pt] node {$0$} (c01)
	(c01) edge [pos=0.25, align=left,below,line width=2.5pt] node {$1$} (c00)
	(c00) edge [pos=0.25, align=left,below] node {$1$} (c02)
	(c02) edge [pos=0.25, align=left,below] node {$2$} (c00)
	(c01) edge [pos=0.25, align=left,below] node {$0$} (c03)
	(c03) edge [pos=0.25, align=left,below] node {$1$} (c01)
	(c02) edge [pos=0.25, align=left,below] node {$1$} (c22)
	(c22) edge [pos=0.25, align=left,below] node {$2$} (c02)
	(c03) edge [pos=0.25, align=left,below] node {$0$} (c31)
	(c31) edge [pos=0.25, align=left,below] node {$1$} (c03)
	(c02) edge [pos=0.25, align=left,left] node {$0$} (c21)
	(c21) edge [pos=0.25, align=left,left] node {$1$} (c02)
	(c02) edge [pos=0.25, align=left,left] node {$0$} (c21)
	(c21) edge [pos=0.25, align=left,left] node {$1$} (c02)
	(c00) edge [pos=0.25, align=left,left] node {$2$} (c11)
	(c11) edge [pos=0.25, align=left,left] node {$3$} (c00)
	(c01) edge [pos=0.25, align=left,left] node {$2$} (c12)
	(c12) edge [pos=0.25, align=left,left] node {$3$} (c01)
	(c03) edge [pos=0.25, align=left,left] node {$2$} (c32)
	(c32) edge [pos=0.25, align=left,left] node {$3$} (c03)
	;
	\node[state] (c00prime) [align=center,line width=0.5mm,fill=gray!30,rectangle,above=of c00] {$C'_0$};
	\node[state] (c01prime) [align=center,fill=gray!30,line width=0.5mm,rectangle, right=of c00prime] {$C'_1$};
	\node[state] (c02prime) [align=center,fill=gray!30, rectangle, left=of c00prime] {$C'_2$};
	\node[state] (c03prime) [align=center,fill=gray!30, rectangle, right=of c01prime] {$C'_3$};
	\node[state] (c11prime) [align=center,fill=black, above=of c00prime] {};
	\node[state] (c12prime) [align=center,fill=black, above=of c01prime] {};
	\node[state] (c22prime) [align=center,fill=green!30!darkgray, left=of c02prime] {};
	\node[state] (c21prime) [align=center,fill=green!30!darkgray, above=of c02prime] {};
	\node[state] (c32prime) [align=center,fill=green!30!darkgray, above=of c03prime] {};
	\node[state] (c31prime) [align=center,fill=green!30!darkgray, right=of c03prime] {};
	\path
	(c00prime) edge [pos=0.25, align=left,below,line width=2.5pt] node {$0$} (c01prime)
	(c01prime) edge [pos=0.25, align=left,below,line width=2.5pt] node {$1$} (c00prime)
	(c00prime) edge [pos=0.25, align=left,below] node {$1$} (c02prime)
	(c02prime) edge [pos=0.25, align=left,below] node {$2$} (c00prime)
	(c01prime) edge [pos=0.25, align=left,below] node {$0$} (c03prime)
	(c03prime) edge [pos=0.25, align=left,below] node {$1$} (c01prime)
	(c02prime) edge [pos=0.25, align=left,below] node {$1$} (c22prime)
	(c22prime) edge [pos=0.25, align=left,below] node {$2$} (c02prime)
	(c03prime) edge [pos=0.25, align=left,below] node {$0$} (c31prime)
	(c31prime) edge [pos=0.25, align=left,below] node {$1$} (c03prime)
	(c02prime) edge [pos=0.25, align=left,left] node {$0$} (c21prime)
	(c21prime) edge [pos=0.25, align=left,left] node {$1$} (c02prime)
	(c02prime) edge [pos=0.25, align=left,left] node {$0$} (c21prime)
	(c21prime) edge [pos=0.25, align=left,left] node {$1$} (c02prime)
	(c00prime) edge [pos=0.25, align=left,left] node {$2$} (c11prime)
	(c11prime) edge [pos=0.25, align=left,left] node {$3$} (c00prime)
	(c01prime) edge [pos=0.25, align=left,left] node {$2$} (c12prime)
	(c12prime) edge [pos=0.25, align=left,left] node {$3$} (c01prime)
	(c03prime) edge [pos=0.25, align=left,left] node {$2$} (c32prime)
	(c32prime) edge [pos=0.25, align=left,left] node {$3$} (c03prime)
	;
	
	\path[red]
	(c00prime) edge [pos=0.25, align=left,left,line width=2.5pt] node {$0$} (c00)
	(c00) edge [pos=0.25, align=left,left,line width=2.5pt] node {$0$} (c00prime)
	(c01prime) edge [pos=0.25, align=left,left,line width=2.5pt] node {$0$} (c01)
	(c01) edge [pos=0.25, align=left,left,line width=2.5pt] node {$0$} (c01prime)
	(c02prime) edge [pos=0.25, align=left,left] node {$0$} (c02)
	(c02) edge [pos=0.25, align=left,left] node {$0$} (c02prime)
	(c03prime) edge [pos=0.25, align=left,left] node {$0$} (c03)
	(c03) edge [pos=0.25, align=left,left] node {$0$} (c03prime)
	(c21prime) edge [pos=0.1, align=left,left,bend right=35] node {$0$} (c21)
	(c21) edge [pos=0.1, align=left,left,bend left=35] node {$0$} (c21prime)
	(c12prime) edge [pos=0.1, align=left,left,bend right=35] node {$0$} (c12)
	(c12) edge [pos=0.1, align=left,left,bend left=35] node {$0$} (c12prime)
	(c32prime) edge [pos=0.1, align=left,left,bend right=35] node {$0$} (c32)
	(c32) edge [pos=0.1, align=left,left,bend left=35] node {$0$} (c32prime)
	(c11prime) edge [pos=0.1, align=left,left,bend right=35] node {$0$} (c11)
	(c11) edge [pos=0.1, align=left,left,bend left=35] node {$0$} (c11prime)
	(c31prime) edge [pos=0.25, align=left,left] node {$0$} (c31)
	(c31) edge [pos=0.25, align=left,left] node {$0$} (c31prime)
	(c22prime) edge [pos=0.25, align=left,left] node {$0$} (c22)
	(c22) edge [pos=0.25, align=left,left] node {$0$} (c22prime)
	;
\end{tikzpicture}